\pdfoutput=1

\documentclass[fleqn,12pt,letter]{article}
\usepackage{natbib}
\usepackage{float,multirow}
\usepackage{amsthm,amsfonts,amsopn,amsmath,amssymb,vmargin,verbatim}
\usepackage{pgf,pgfplots,tikz}
\usetikzlibrary{arrows, calc, fit, matrix, positioning, shapes.multipart,shapes.symbols,svg.path}
 
\usepackage[onehalfspacing]{setspace}

\usepackage[toc,page]{appendix}

\RequirePackage[pagebackref=true]{hyperref}
\hypersetup{colorlinks=true,linkcolor=blue,urlcolor=blue,citecolor=red,
    pdftitle=Identifying the Discount Factor in Dynamic Discrete Choice Models,
    pdfauthor=Jaap Abbring and Oeystein Daljord,
    pdfsubject=JEL Codes C25 C61,
    pdfkeywords=discount factor dynamic discrete choice empirical content identification 
    pdfdisplaydoctitle=true}

\newcommand{\cites}[1]{\citeauthor{#1}'s (\citeyear{#1})}
\theoremstyle{definition}
\newtheorem{example}{Example}

\newtheorem{theorem}{Theorem}
\newtheorem{corollary}{Corollary}

\def\figwidth{12cm}
\def\figheight{7.5cm}

\def\moreEmpiricalContentK{ 2}
\def\moreEmpiricalContentJ{ 3}

\def\moreEmpiricalContentMinRhs{-0.1278}
\def\moreEmpiricalContentPreciseLhs{0.0800}

\def\moreEmpiricalContentPreciseBetaMT{  0.72}
\def\moreEmpiricalContentLhs{0.0800}

\def\moreEmpiricalContentBetaMT{  0.72}
\def\moreEmpiricalContentQiRowOne{\left[\begin{array}{ccc}  0.25&  0.25&  0.50\end{array}\right]}
\def\moreEmpiricalContentQiRowTwo{\left[\begin{array}{ccc}  0.00&  0.25&  0.75\end{array}\right]}
\def\moreEmpiricalContentQK{\left[\begin{array}{ccc}  0.90&  0.00&  0.10\\  0.00&  0.90&  0.10\\  0.00&  1.00&  0.00\end{array}\right]}
\def\moreEmpiricalContentPOne{\left[\begin{array}{ccc}  0.50\\  0.48\\  0.10\end{array}\right]}
\def\moreEmpiricalContentPK{\left[\begin{array}{ccc}  0.50\\  0.52\\  0.90\end{array}\right]}
\def\moreEmpiricalContentMPrime{\left[\begin{array}{ccc}  0.69&  0.65&  0.11\end{array}\right]}
\def\moreEmpiricalContentDeltaQ{\left[\begin{array}{ccc} -0.65&  0.90& -0.25\end{array}\right]}

\def\moreEmpiricalContentDeltaQM{0.1116}

\def\needNoMTK{ 2}
\def\needNoMTJ{ 3}

\def\needNoMTPreciseLhs{0.0800}
\def\needNoMTPreciseBetaOne{  0.90}

\def\needNoMTLhs{0.0800}
\def\needNoMTBetaOne{  0.90}

\def\needNoMTQiRowOne{\left[\begin{array}{ccc}  0.00&  0.25&  0.75\end{array}\right]}
\def\needNoMTQiRowTwo{\left[\begin{array}{ccc}  0.25&  0.25&  0.50\end{array}\right]}
\def\needNoMTQK{\left[\begin{array}{ccc}  0.00&  1.00&  0.00\\  0.00&  1.00&  0.00\\  0.00&  0.00&  1.00\end{array}\right]}
\def\needNoMTPOne{\left[\begin{array}{ccc}  0.50\\  0.48\\  0.50\end{array}\right]}
\def\needNoMTPK{\left[\begin{array}{ccc}  0.50\\  0.52\\  0.50\end{array}\right]}
\def\needNoMTMPrime{\left[\begin{array}{ccc}  0.69&  0.65&  0.69\end{array}\right]}
\def\needNoMTDeltaQ{\left[\begin{array}{ccc} -0.25&  0.00&  0.25\end{array}\right]}

\def\needNoMTDeltaQM{0.0000}

\def\identificationFailsK{ 2}
\def\identificationFailsJ{ 3}

\def\identificationFailsPreciseLhs{0.0400}
\def\identificationFailsPreciseBetaOne{  0.34}
\def\identificationFailsPreciseBetaTwo{  0.95}
\def\identificationFailsPreciseBetaMT{  0.31}
\def\identificationFailsLhs{0.0400}
\def\identificationFailsBetaOne{  0.34}
\def\identificationFailsBetaTwo{  0.95}
\def\identificationFailsBetaMT{  0.31}
\def\identificationFailsQiRowOne{\left[\begin{array}{ccc}  0.25&  0.25&  0.50\end{array}\right]}
\def\identificationFailsQiRowTwo{\left[\begin{array}{ccc}  0.00&  0.25&  0.75\end{array}\right]}
\def\identificationFailsQK{\left[\begin{array}{ccc}  0.90&  0.00&  0.10\\  0.00&  0.90&  0.10\\  0.00&  1.00&  0.00\end{array}\right]}
\def\identificationFailsPOne{\left[\begin{array}{ccc}  0.50\\  0.49\\  0.10\end{array}\right]}
\def\identificationFailsPK{\left[\begin{array}{ccc}  0.50\\  0.51\\  0.90\end{array}\right]}
\def\identificationFailsMPrime{\left[\begin{array}{ccc}  0.69&  0.67&  0.11\end{array}\right]}
\def\identificationFailsDeltaQ{\left[\begin{array}{ccc} -0.65&  0.90& -0.25\end{array}\right]}

\def\identificationFailsDeltaQM{0.1291}

\def\twiceTwoK{ 2}
\def\twiceTwoJ{ 4}

\def\twiceTwoOnePreciseLhs{0.0068}
\def\twiceTwoTwoPreciseLhs{0.0019}
\def\twiceTwoOnePreciseBetaOne{  0.17}
\def\twiceTwoTwoPreciseBetaOne{  0.07}
\def\twiceTwoOnePreciseBetaTwo{  0.90}

\def\twiceTwoOneLhs{0.0068}
\def\twiceTwoTwoLhs{0.0019}
\def\twiceTwoOneBetaOne{  0.17}
\def\twiceTwoTwoBetaOne{  0.07}
\def\twiceTwoOneBetaTwo{  0.90}
\def\twiceTwoTwoBetaTwo{  0.90}

\def\twiceTwoQi{\left[\begin{array}{cccc}  0.43&  0.26&  0.18&  0.18\\  0.33&  0.29&  0.36&  0.27\\  0.19&  0.26&  0.18&  0.45\\  0.05&  0.18&  0.29&  0.09\end{array}\right]}
\def\twiceTwoQK{\left[\begin{array}{cccc}  0.17&  0.26&  0.13&  0.43\\  0.13&  0.07&  0.20&  0.60\\  0.20&  0.30&  0.10&  0.40\\  0.25&  0.15&  0.50&  0.10\end{array}\right]}
\def\twiceTwoPOne{\left[\begin{array}{cccc}  0.92&  0.92&  0.63&  0.63\end{array}\right]}
\def\twiceTwoPK{\left[\begin{array}{cccc}  0.08&  0.08&  0.37&  0.37\end{array}\right]}

\def\twiceTwoNoiseDistScaler{ 4}

\def\twiceTwoNoiseOnePreciseLhs{0.0066}
\def\twiceTwoNoiseTwoPreciseLhs{0.0050}
\def\twiceTwoNoiseOnePreciseBetaOne{  0.16}
\def\twiceTwoNoiseTwoPreciseBetaOne{  0.25}
\def\twiceTwoNoiseOnePreciseBetaTwo{  0.91}
\def\twiceTwoNoiseTwoPreciseBetaTwo{  0.68}

\def\twiceTwoNoiseOneLhs{0.0066}
\def\twiceTwoNoiseTwoLhs{0.0050}
\def\twiceTwoNoiseOneBetaOne{  0.16}
\def\twiceTwoNoiseTwoBetaOne{  0.25}
\def\twiceTwoNoiseOneBetaTwo{  0.91}
\def\twiceTwoNoiseTwoBetaTwo{  0.68}

\def\twiceTwoNoisePreciseCriticalValue{0.1000}
\def\twiceTwoNoisePreciseBetaSetOne{  0.10}
\def\twiceTwoNoisePreciseBetaSetTwo{  0.28}
\def\twiceTwoNoisePreciseBetaSetThree{  0.79}
\def\twiceTwoNoisePreciseBetaSetFour{  0.91}
\def\twiceTwoNoiseCriticalValue{  0.10}
\def\twiceTwoNoiseBetaSetOne{  0.10}
\def\twiceTwoNoiseBetaSetTwo{  0.28}
\def\twiceTwoNoiseBetaSetThree{  0.79}
\def\twiceTwoNoiseBetaSetFour{  0.91}

\def\oneRedundantK{ 2}
\def\oneRedundantJ{ 4}

\def\oneRedundantOnePreciseLhs{0.0187}
\def\oneRedundantTwoPreciseLhs{0.0045}
\def\oneRedundantOnePreciseBetaOne{  0.30}

\def\oneRedundantTwoPreciseBetaTwo{  0.65}

\def\oneRedundantOneLhs{0.0187}
\def\oneRedundantTwoLhs{0.0045}
\def\oneRedundantOneBetaOne{  0.30}

\def\oneRedundantTwoBetaTwo{  0.65}

\def\oneRedundantQi{\left[\begin{array}{cccc}  0.40&  0.26&  0.18&  0.18\\  0.33&  0.29&  0.36&  0.27\\  0.19&  0.26&  0.18&  0.45\\  0.08&  0.18&  0.29&  0.09\end{array}\right]}
\def\oneRedundantQK{\left[\begin{array}{cccc}  0.17&  0.26&  0.13&  0.43\\  0.13&  0.07&  0.20&  0.60\\  0.20&  0.30&  0.10&  0.40\\  0.25&  0.15&  0.50&  0.10\end{array}\right]}
\def\oneRedundantPOne{\left[\begin{array}{cccc}  0.60&  0.59&  0.88&  0.88\end{array}\right]}
\def\oneRedundantPK{\left[\begin{array}{cccc}  0.40&  0.41&  0.12&  0.12\end{array}\right]}

\def\experienceK{ 2}
\def\experienceJ{ 3}

\def\experiencePreciseLhs{0.4918}
\def\experiencePreciseBetaOne{  0.80}

\def\experienceLhs{0.4918}
\def\experienceBetaOne{  0.80}

\def\experienceQiRowOne{\left[\begin{array}{ccc}  0.00&  0.25&  0.75\end{array}\right]}
\def\experienceQiRowTwo{\left[\begin{array}{ccc}  0.25&  0.75&  0.00\end{array}\right]}
\def\experienceQK{\left[\begin{array}{ccc}  1.00&  0.00&  0.00\\  0.50&  0.50&  0.00\\  0.00&  0.50&  0.50\end{array}\right]}
\def\experiencePOne{\left[\begin{array}{ccc}  0.44\\  0.56\\  0.71\end{array}\right]}
\def\experiencePK{\left[\begin{array}{ccc}  0.56\\  0.44\\  0.29\end{array}\right]}
\def\experienceMPrime{\left[\begin{array}{ccc}  0.57&  0.82&  1.23\end{array}\right]}
\def\experienceDeltaQ{\left[\begin{array}{ccc}  0.25& -1.00&  0.75\end{array}\right]}
\def\experienceInlineDeltaQ{$[0.25$~$-1.00$~$0.75$]}
\def\experienceDeltaQM{0.2465}

\def\experienceStructurelearnRate{  0.75}
\def\experienceStructuredepRate{  0.50}
\def\experienceStructurebeta{  0.80}

\def\experienceStructureuOne{-0.50}

\def\experienceStructureuThree{0.50}

\begin{document}

\title{
Identifying the Discount Factor in Dynamic Discrete Choice Models\protect\thanks{We thank Nikhil Agarwal, \'{A}ureo de Paula, Jean-Pierre Dub\'{e}, Hanming Fang, Christian Hansen, G\"{u}nter Hitsch, John-Eric Humphries, Robert Miller, Whitney Newey, John Rust, Azeem Shaikh, Eduardo Souza-Rodrigues, Elie Tamer, Michela Tincani, Thomas Wollmann, the editor (Christopher Taber), anonymous referees, and participants in the 2016 Cowles Commission Summer Conference, the 69th European Meeting of the Econometric Society, the 2nd Conference on Structural Dynamic Models in Copenhagen, the BRIQ Workshop on Structural Analysis of Inequality in Bonn, and seminars at Tilburg, Toronto, Copenhagen, Chicago, UCL, LSE, Penn State, DIW Berlin, Georgetown, EIEF Rome, Cambridge, Warwick, Bonn, Z\"{u}rich, UC Irvine, CalTech, and USC for helpful discussions and comments. Jaap Abbring's research is financially supported by the Netherlands Organisation for Scientific Research (NWO) through Vici grant 453-11-002. \O ystein Daljord gratefully acknowledges financial support from the University of Chicago Booth School of Business. }}

\author{Jaap H. Abbring\thanks{CentER, Department of Econometrics \& OR, Tilburg University, P.O. Box 90153, 5000 LE Tilburg, The Netherlands; and CEPR. E-mail: \href{mailto:jaap@abbring.org}{jaap@abbring.org}. Web: \href{http://jaap.abbring.org}{jaap.abbring.org}.}\and
\O ystein Daljord\thanks{Booth School of Business, University of Chicago, 5807 South Woodlawn Avenue, Chicago, IL 60637, USA. E-mail:
\href{mailto:Oeystein.Daljord@chicagobooth.edu}{Oeystein.Daljord@chicagobooth.edu}. Web: \href{http://faculty.chicagobooth.edu/oystein.daljord}{faculty.chicagobooth.edu/oystein.daljord}.
}}

\date{September 2019}

\maketitle

\vspace*{-7mm}
\begin{abstract}
Empirical research often cites observed choice responses to variation that shifts expected discounted future utilities, but not current utilities, as an intuitive source of information on time preferences. We study the identification of dynamic discrete choice models under such economically motivated exclusion restrictions on primitive utilities. We show that each exclusion restriction leads to an easily interpretable moment condition with the discount factor as the only unknown parameter. The identified set of discount factors that solves this condition is finite, but not necessarily a singleton. Consequently, in contrast to common intuition, an exclusion restriction does not in general give point identification. Finally, we show that exclusion restrictions have nontrivial empirical content: The implied moment conditions impose restrictions on choices that are absent from the unconstrained model.
\end{abstract}

\thispagestyle{empty}

\clearpage

\section{Introduction}

Identification of the discount factor in dynamic discrete choice models is crucial for their application to the evaluation of agents' responses to dynamic interventions. It is, however, well known that the discount factor is not identified from choice data without further restrictions (\citealp{nh94:rust}, Lemma 3.3, and  \citealp{ecta02:magnacthesmar}, Proposition 2).  Consequently, empirical researchers usually fix the discount factor at some a priori plausible value, e.g. 0.95, or impose ad hoc functional form assumptions that allow it to be identified and estimated. These approaches solve the identification problem, but often lack economic justification. Inferring the time preferences in the specific context of an application is important as discount factors have been estimated to vary substantially across choice contexts and populations \citep{fredericketal02}.\footnote{\citeauthor{fredericketal02} also showed that geometric discounting is often rejected in data in favor of present biased time preferences. We study the identification and estimation of hyperbolic discount functions in \cite{abbringdaljordiskhakov18}.}  

In this paper, we explore identification from observed choice responses to variation that shifts expected discounted future utilities, but not current utilities. Such variation is commonly cited in applications as an intuitive source of information on time preferences. For example, in studies of green technology adoption, \cite{bollinger15}  and \cite{aer19:degrooteverboven} argued that firms' and households' current choice responses to regulation that shifts their future expenses, but not their current expenses, are informative about discount factors. In a study of demand for game consoles, \cite{Lee2013} assumed that the discount factor is identified from variation in the expected quality of future releases, which shifts future values without affecting current payoffs.  In an application to cellphone plan choice, \cite{yaoetal12} argues informally that utilities can be identified in a terminal period when the choice problem is static. The discount factor can subsequently be identified from choices in the next to last period.   \cite{chungetal14} appeals to \cite{yaoetal12}'s idea in a study of salesforce compensation plans. We give further examples from the literature in Section \ref{s:examples}.

In Section \ref{s:new}, we formalize the intuition in these studies as an exclusion restriction on primitive utilities. We first consider a stationary model with infinite horizon (introduced in Section \ref{s:model}). We prove that, in contrast to common intuition, an exclusion restriction on primitive utilities does not generally point identify the discount factor. It does however narrow the identified set--- the set of observationally equivalent discount factors--- to a discrete and, if we exclude values near one, finite set. This set contains the solutions to a moment condition that only involves the discount factor and that has a straightforward interpretation in terms of choice responses to variation in expected discounted future utilities. The moment condition can be used directly in estimation, independently of the rest of the model parameters. 

We subsequently provide a finite upper bound on  the number of discount factors in the identified set for the case in which the states display finite dependence, as defined by \cite{arcidiaconomiller11, arcidiaconomillerlongshort}. Examples include optimal stopping and renewal problems, which we show to be point identified. 

We extend our analysis to nonstationary models with finite horizons, which are commonly used in labor applications (\citealp{res89:ecksteinwolpin} and \citealp{jpe97:keanewolpin} are early examples). We show that, with exclusion restrictions, the discount factor is generally identified up to a finite set in these models. 

In Section \ref{s:empiricalcontent}, we explore the empirical content of exclusion restrictions. \citeauthor{ecta02:magnacthesmar}'s Proposition 2 implies that dynamic discrete choice models without exclusion restrictions cannot be falsified with data on choices and states. In that sense, the models have no empirical content. We show that exclusion restrictions impose nontrivial restrictions on the data, which can be tested. 

Finally, in Section \ref{s:multiple}, we argue that common intuition often supports {\em multiple} exclusion restrictions, which imply multiple moment conditions. These moment conditions share the true discount factor (if one exists that rationalizes the data) as one solution, but may have individually more solutions. We discuss how standard (set) estimators can be applied to this case. 

 This paper's main contribution is to provide a simple and intuitively appealing analysis of identification of the discount factor in dynamic discrete choice models under economically motivated exclusion restrictions. Our analysis complements a substantial literature in econometrics (see \citealp{nh94:rust} and \citealp{are10:abbring}, for reviews). \citeauthor{ecta02:magnacthesmar}'s Proposition 4 established point identification based on a different type of exclusion restriction than ours: the existence of a pair of states that affects, in some specific way, expected discounted future utilities, but not the ``current value,'' which is a difference in expected discounted utilities between two particular choice sequences. This is a high level exclusion restriction that is difficult to interpret and hard to verify in applications. In particular, unlike our exclusion restriction, it does not formalize the common intuition that is given in applications like those discussed above. Empirical applications often incorrectly cite \citeauthor{ecta02:magnacthesmar}'s result as one for an exclusion restriction on primitive utility. For example, in a study of housing location choice, \citet[][p. 921]{ecta16:bayeretal} wrote
\begin{quote}
Magnac and Thesmar (2002) \ldots showed that dynamic models are identified with an appropriate exclusion restriction--- in particular, a variable that shifts expectations but not current utility. In the context  of  our  framework,  lagged  amenities  provide  exactly  this  sort  of  exclusion  restriction:  while  current  utility  depends  on  the  current  level  of  the amenities provided in a neighborhood, lagged amenity levels help predict how amenities will evolve going forward and thus contribute to expectations about the future utility associated with that choice of neighborhood.
\end{quote}
We show how \citeauthor{ecta16:bayeretal}'s exclusion restriction can be used to set identify and estimate the discount factor, even if it is insufficient for point identification.\footnote{We thank John Rust for this example.} 

\citeauthor{ecta02:magnacthesmar}'s identification result relies on a rank condition that ensures sufficient variation in expected discounted future utilities. This rank condition does not suffice for point identification with our exclusion restriction on primitive utilities. We do however use natural extensions of this condition to ensure local identification of myopic preferences, which is needed for our discrete set identification result. 

\citeauthor{ecta02:magnacthesmar}'s Proposition 2 implies that, without further restrictions, not only the discount factor, but also the utility of one reference choice can be normalized without restricting the observed choice and transition probabilities. Intuitively, discrete choices only identify utility contrasts, not levels. However,  {\em counterfactual} choice probabilities, which are often the objects of interest in dynamic discrete choice analysis, are generally not invariant to the choice of reference utility \citep{res14:noretstang,kalouptsidietal16}. This suggests that we do not only treat the discount factor, but also the utility of the reference choice as a free parameter that should be determined from data. Indeed, we view the identification of the reference utility as an important, but separate problem from the identification of the discount factor. For expositional convenience, we derive our main results under the normalization that the reference utility equals zero. In the appendix, we show that our results straightforwardly extend to the case in which the reference utility is known up to a constant shift. 

We emphasize that the idea of using exclusion restrictions to identify time preferences in choice models is not ours, but has circulated in the literature for a while. One early example is \cite{chevaliergoolsbee09}, which studied demand for textbooks. Its choice model implicitly excluded the expected future resale price of a textbook from the current period pay-off to identify a discount factor.  \cite{ier15:fangwang} explicitly proposed the use of exclusion restrictions similar to ours to identify a dynamic discrete choice model with partially naive hyperbolic time preferences. In \citet{fw19:abbringdaljord}, we argue that \citeauthor{ier15:fangwang}'s main generic identification result has no implications for the identification of the model with hyperbolic discounting or its special case with geometric discounting, which we study. In any case, our approach is different: We isolate the specific empirical implications of the exclusion restrictions, whereas \citeauthor{ier15:fangwang} studied their model as a general system of nonlinear equations, using results from differential topology.

\cite{qe18:komarovaetal} showed point identification of the discount factor under parametric assumptions on the utility function in a model like ours, but without exclusion restrictions. \citeauthor{res14:noretstang} demonstrated that in a model with parametric utility, point identification is lost to set identification when the distribution of unobservables is allowed to deviate from a known one, such as the type-1 extreme value specification that underlies logit choice probabilities. Without any restrictions on the distribution of unobservables beyond conditional independence and absolute continuity, all identification of the discount factor is lost, i.e. the identified set of discount factors is the unit interval.   We instead focus on identification for a nonparametric utility function under economically motivated exclusion restrictions.  We map each exclusion restriction to an easily interpretable and computable moment condition that directly informs the identification and estimation of the discount factor, and the model's empirical content. 

\section{Model}
\label{s:model}

Consider a stationary dynamic discrete choice model \citep[e.g.][]{nh94:rust}. Time is discrete with an infinite horizon.\footnote{Section \ref{ss:finitehorizon} considers an extension to a nonstationary model with a finite horizon.} In each period, agents first observe state variables $x$ and $\varepsilon$, where $x$ takes discrete values in ${\cal X}=\{x_1,\ldots,x_J\}$ and $\varepsilon=\{\varepsilon_1,\ldots,\varepsilon_K\}$ is continuously distributed on $\mathbb{R}^K$; for $J,K\geq 2$. Then, they choose $d$ from the set of alternatives ${\cal D}=\{1,2,\ldots,K\}$ and collect utility $u_d(x,\varepsilon)=u_d^*(x)+\varepsilon_d$. Finally,  they move to the next period with new state variables $x'$ and $\varepsilon'$ drawn from a Markov transition distribution controlled by $d$. We assume that a version of \cites{ecta87:rust} conditional independence assumption holds. Specifically, $x'$ is drawn independently of $\varepsilon$ from the  transition distribution $Q_k\left(\cdot|x\right)$ for any choice $k \in {\cal D}$; and  $\varepsilon_1,\ldots,\varepsilon_K$ are independently drawn from mean zero type-1 extreme value distributions.\footnote{\citeauthor{ecta02:magnacthesmar} showed that the distribution of $\varepsilon$ cannot be identified and took it to be known. Our type-1 extreme value assumption leads to the canonical multinomial logit case. Our results extend directly to any other known, continuous distribution on $\mathbb{R}^K$.}  Agents maximize the rationally expected utility flow discounted with factor $\beta\in[0,1)$.

Each choice $d$ equals the option $k$ that maximizes the choice-specific  expected discounted utility (or, simply, ``value'') $v_k(x,\varepsilon)$. The additive separability of $u_k(x,\varepsilon)$ and conditional independence imply that $v_k(x,\varepsilon)=v^*_k(x)+\varepsilon_k$, with $v^*_k$ the unique solution to
\begin{equation}
\label{eq:bellman}
\begin{split}
v^*_k(x)&=u^*_k(x) +\beta  \mathbb{E}\left[\max_{{k'}\in{\cal D}}\{v^*_{k'}(x')+\varepsilon'_{k'}\} \;\vline\; d=k,x\right]\\
&=u^*_k(x) + \beta\int \mathbb{E}\left[\max_{{k'}\in{\cal D}}\{v^*_{k'}( \tilde x)+\varepsilon'_{k'}\}\right]dQ_k(\tilde x|x)
\end{split}
\end{equation}
for all $k\in {\cal D}$. Here, for each given $\tilde{x}\in {\cal X}$, 
\begin{equation}
\label{eq:McFadden}
	\mathbb{E}\left[\max_{{k'}\in{\cal D}}\{v^*_{k'}(\tilde{x})+\varepsilon'_{k'}\}\right]=\ln\left(\sum_{k'\in{\cal D}} \exp\left(v^*_{k'}(\tilde{x})\right)\right)
\end{equation} 
is the McFadden surplus for the choice among $k'\in{\cal D}$ with utilities $v^*_{k'}(\tilde{x})+\varepsilon'_{k'}$.

Suppose we have data on choices $d$ and state variables $x$ that allow us to determine $Q_k(\cdot|\tilde{x})$ and the choice probabilities $p_k(\tilde{x})=\Pr(d=k | x=\tilde{x})$ for all $k\in{\cal D}$ and $\tilde{x}\in{\cal X}$. The model is point identified if and only if we can uniquely determine its primitives from these data. As we discuss in Section \ref{s:empiricalcontent}, a version of \citeauthor{ecta02:magnacthesmar}'s Proposition 2 holds: There exist unique (up to a standard utility normalization) values of the primitives that rationalize the data for any given discount factor $\beta\in[0,1)$. We therefore focus our identification analysis on $\beta$.

The choice probabilities are fully determined by
\begin{equation}
\label{eq:choiceprobinversion}
 \ln\left(p_k(\tilde{x}) \right) - \ln\left(p_K(\tilde{x})\right)=v^*_k(\tilde{x})- v^*_K(\tilde{x}),~~~~~~k\in{\cal D}/\{K\},~\tilde{x}\in{\cal X}.
\end{equation}
The transition probabilities $Q_k(\cdot|\tilde{x})$, the value contrasts $v^*_k(\tilde{x})- v^*_K(\tilde{x})$ for $k\in{\cal D}/\{K\}$ and $\tilde{x}\in{\cal X}$ therefore capture all the model's implications for the data.  \citet{res93:hotzmiller} pointed out that \eqref{eq:choiceprobinversion} can be inverted to identify the value contrasts from the choice probabilities. To use this, we first rewrite (\ref{eq:bellman}) as
\begin{equation}
\label{eq:bellmanprobs}
v^*_k(x)=u^*_k(x)+\beta\int \left(m(\tilde x)+v^*_K(\tilde x)\right) dQ_k(\tilde x|x),
\end{equation}
where, for given $\tilde{x}\in{\cal X}$, $m(\tilde{x})=\mathbb{E}\left[\max_{{k'}\in{\cal D}}\{v^*_{k'}(\tilde{x})-v^*_{K}(\tilde{x})+\varepsilon'_{k'}\}\right]$ is the ``excess  surplus'' (over $v^*_K(\tilde{x})$), the McFadden surplus for the choice among $k'\in{\cal D}$ with utilities $v^*_{k'}(\tilde{x})-v^*_{K}(\tilde{x})+\varepsilon'_{k'}$. By \eqref{eq:McFadden}  and \eqref{eq:choiceprobinversion}, it follows that $m(\tilde{x})=-\ln \left(p_K(\tilde{x})\right)$.

\section{Identification}
\label{s:new}

Let $\mathbf{v}_k$, $\mathbf{p}_k$, $\mathbf{u}_k$, and $\mathbf{m}$ be $J\times 1$ vectors with $j$-th elements $v^*_k(x_j)$, $p_k(x_j)$, $u^*_k(x_j)$, and $m(x_j)$, respectively. Let $\mathbf{Q}_k$ be the $J\times J$ matrix with $(j,j')$-th entry $Q_k(x_{j'}|x_j)$ and $\mathbf{I}$ be a $J\times J$ identity matrix. Note that the $J\times 1$ vector $\mathbf{m}  + \mathbf{v}_K$ stacks the McFadden surpluses in \eqref{eq:McFadden}.  In this notation, the data are $\{\mathbf{p}_k,\mathbf{Q}_k; k\in{\cal D}\}$ and directly identify $\mathbf{m}=-\ln\mathbf p_K$ \citep[][Lemma 1 and Section 3.3]{arcidiaconomiller11}.  

\subsection{Magnac and Thesmar's result}
\label{ss:magnacthesmar}

We can rewrite (\ref{eq:bellmanprobs})  as $v^*_k(x)=u^*_k(x)+\beta\mathbf{Q}_k(x)\left[\mathbf{m}+\mathbf{v}_K\right]$, where $\mathbf{Q}_k(x_j)$ is the $j$-th row of $\mathbf{Q}_k$. Subtracting the same expression for $v^*_K(x)$, rearranging, and substituting (\ref{eq:choiceprobinversion}), we get
\begin{equation}
\label{eq:uMT}
\ln (p_k(x))-\ln(p_K(x))=\beta\left[\mathbf{Q}_k(x)-\mathbf{Q}_K(x)\right]\mathbf{m}+ U_k(x),
\end{equation}
where $U_k(x)=u^*_k(x) -u_K^*(x) + \beta\left[\mathbf{Q}_k(x)-\mathbf{Q}_K(x)\right]\mathbf{v}_K$ is \citeauthor{ecta02:magnacthesmar}'s  ``current value" of choice $k$ in state $x$. Its Proposition 4 assumes the existence of a known option $k\in{\cal D}/\{K\}$ and a known pair of states $\tilde x_1,\tilde x_2\in {\cal X}$ such that $\tilde{x}_1\not =\tilde{x}_2$ and $U_k(\tilde x_1)=U_k(\tilde x_2)$. Under this exclusion restriction, differencing (\ref{eq:uMT}) evaluated at $\tilde x_1$ and $\tilde x_2$ yields
\begin{equation}
\label{eq:diffMT}
\begin{split}
\ln\left(p_k(\tilde x_1)/p_K(\tilde x_1)\right) - \ln\left(p_k(\tilde x_2)/p_K(\tilde x_2)\right)&\\
&\hspace*{-14em}= \beta\left[\mathbf{Q}_k(\tilde x_1)-\mathbf{Q}_K(\tilde x_1)- \mathbf{Q}_k(\tilde x_2)+\mathbf{Q}_K(\tilde x_2)\right]\mathbf{m}.
\end{split}
\end{equation}
Given the choice and transition probabilities, the left hand side of \eqref{eq:diffMT} is a known scalar and its right hand side is a known linear function of $\beta$. Therefore, provided that \citeauthor{ecta02:magnacthesmar}'s rank condition 
\begin{equation}
\label{eq:rankMT}
\left[\mathbf{Q}_k(\tilde x_1)-\mathbf{Q}_K(\tilde x_1)- \mathbf{Q}_k(\tilde x_2)+\mathbf{Q}_K(\tilde x_2)\right]\mathbf{m}\not = 0
\end{equation} 
holds, moment condition \eqref{eq:diffMT}  uniquely determines $\beta$ in terms of the choice data. 

This identification argument can be interpreted in terms of an experiment that shifts the expected excess surplus contrast  $\left[\mathbf{Q}_k(x) - \mathbf{Q}_K(x)\right]\mathbf{m}$ by changing the state $x$ from $\tilde{x}_2$ to $\tilde{x}_1$, while keeping the current value $U_k(\tilde{x}_1)=U_k(\tilde{x}_2)$ constant. The discount factor is the per unit effect of that observed shift on the observed log choice probability ratio $\ln\left(p_k(x)/p_K(x)\right)$. 

A shift in the expectation contrast $\mathbf{Q}_k(x) - \mathbf{Q}_K(x)$ does not suffice for identification. For example, suppose that the exclusion restriction holds for some $\tilde{x}_1, \tilde{x}_2\in {\cal X}$, but that the excess surplus $m(x_1) = \cdots=m(x_J)$ is constant, so that the expected excess surplus contrast $\left[\mathbf{Q}_k(x) - \mathbf{Q}_K(x)\right]\mathbf{m}=0$. Then, a shift in the expectation contrast does not shift the expected excess surplus contrast and hence does not change the decision problem. Consequently, this shift is not informative on $\beta$ and \citeauthor{ecta02:magnacthesmar}'s rank condition \eqref{eq:rankMT} fails. 

Rank condition \eqref{eq:rankMT} has a meaningful interpretation and is verifiable in data. The exclusion restriction $U_k(\tilde x_1)=U_k(\tilde x_2)$, however, is more problematic, because it imposes opaque conditions on the primitives that are hard to verify in applications. The current values depend on both current utilities and discounted expected future values. Specifically, they involve elements of $\mathbf{v}_K$,  which by (\ref{eq:bellmanprobs}) equals 
\begin{equation}\label{eq:vK}
\mathbf{v}_K=\left[\mathbf{I}-\beta \mathbf{Q}_K\right]^{-1}\left[\mathbf{u}_K+\beta\mathbf{Q}_K\mathbf{m}\right]. 
\end{equation}
The current value is in fact a value contrast between two sequences of choices: choose $k$ now, $K$ in the next period, and choose optimally ever after, relative to choose $K$ now, $K$ in the next period, and choose optimally ever after.  Because this particular value contrast does not correspond to common economic choice sequences, the applied value of \citeauthor{ecta02:magnacthesmar}'s restriction is limited. It is hard to think of naturally occurring experiments that shift the expected contrasts in excess surplus, i.e. satisfy the rank condition, without also shifting the current value and consequently violating the exclusion restriction, except for special cases.  Indeed, the intuitive identification arguments in the introduction's empirical examples do not involve current values, but exclusion restrictions on primitive utility.

\subsection{An exclusion restriction on primitive utility}

Like \citeauthor{ecta02:magnacthesmar}, we start with (\ref{eq:uMT}) or, equivalently, 
\begin{equation}
\label{eq:uMTnew}
\ln \mathbf{p}_k - \ln \mathbf{p}_K
=  \beta\left[\mathbf{Q}_k-\mathbf{Q}_K\right]\left[\mathbf{m}+\mathbf{v}_K\right] + \mathbf{u}_k -\mathbf{u}_K.
\end{equation}
Instead of controlling the contribution of $\mathbf{v}_K$ to the right hand side with an exclusion restriction on the current value, we exploit that it can be expressed in terms of the model primitives. Substituting (\ref{eq:vK}) in (\ref{eq:uMTnew}) and rearranging gives 
\begin{equation}
\label{eq:u}
\ln \mathbf{p}_k - \ln \mathbf{p}_K
=  \beta\left[\mathbf{Q}_k-\mathbf{Q}_K\right]\left[\mathbf{I}-\beta \mathbf{Q}_K\right]^{-1}\left[\mathbf{m}+\mathbf{u}_K\right] + \mathbf{u}_k-\mathbf{u}_K.
\end{equation}
Intuition from static discrete choice analysis and \citeauthor{ecta02:magnacthesmar}'s results for dynamic models suggest that, for identification, we need to fix utility in one reference alternative, say $\mathbf{u}_K$. Intuitively, choices only depend on, and thus inform about, utility contrasts. Thus,  following e.g. \citeauthor{ier15:fangwang} and \citet{nber15:bajarietal}, we set $\mathbf{u}_K=\mathbf{0}$.\footnote{Note that this normalization does not collapse \citeauthor{ecta02:magnacthesmar}'s exclusion restriction on current values to an easily interpretable restriction on primitives.} This normalization cannot be refuted by data without further restrictions (see Section \ref{s:empiricalcontent}). Despite this lack of empirical content, it is not completely innocuous, as it may affect the model's counterfactual predictions (see e.g. \citeauthor{res14:noretstang}, Lemma 2, and \citeauthor{kalouptsidietal16}). In the appendix, we demonstrate that our analysis applies without change to the case in which $u^*_K(x)$ is constant, but not necessarily zero, and can straightforwardly be extended to the case in which $u^*_K(x)$ is known up to a constant shift, but not necessarily constant. Thus, our analysis of the identification of the discount factor complements identification results for the reference utility $u^*_K$.\footnote{\citet{usc15:chou} recently provided identification results for dynamic discrete choice models without a normalization of $u^*_K$. \citeauthor{usc15:chou}'s results for the stationary model that we study here take the discount factor to be known. \citeauthor{usc15:chou}'s Propositions 3, 7, and 8 for a nonstationary model like the one we study in Section \ref{ss:finitehorizon} provide high-level sufficient conditions for point identification, whereas we focus on set identification under intuitive conditions. A general difference is that we emphasize the economic interpretation of the identifying conditions and that we provide results on their empirical content.} 

Now suppose that we know the value of $u_k^*(\tilde x_1)-u_l^*(\tilde x_2)$ for some known choices $k\in{\cal D}/\{K\}$ and $l \in {\cal D}$ and known states $\tilde x_1 \in {\cal X}$ and $\tilde x_2\in{\cal X}$; with either $k\neq l$,  $\tilde{x}_1\not =\tilde{x}_2$, or both. For expositional convenience only (see the appendix for the general case), we take this known value to be zero, and simply focus on the exclusion restriction
\begin{align}\label{eq:exclusionrestriction}
u_k^*(\tilde x_1)=u_l^*(\tilde x_2).
\end{align}
An advantage of this exclusion restriction over \citeauthor{ecta02:magnacthesmar}'s current value restriction is that it is a direct constraint on primitive utility with a clear economic interpretation. It also extends \citeauthor{ecta02:magnacthesmar} by allowing for restrictions on primitive utilities across combinations of choices and states.

\subsection{The identified set}

Under exclusion restriction \eqref{eq:exclusionrestriction}, we can difference \eqref{eq:u} to implicitly relate $\beta$ to the choice data (the choice and transition probabilities), without  reference to any other unknown parameters (the utilities):
\begin{equation}
\label{eq:diffAD}
\begin{split}
\ln\left(p_k(\tilde x_1)/p_K(\tilde x_1)\right) - \ln\left(p_l(\tilde x_2)/p_K(\tilde x_2)\right)&\\
&\hspace*{-12em}= \beta\left[\mathbf{Q}_k(\tilde x_1)-\mathbf{Q}_K(\tilde x_1)- \mathbf{Q}_l(\tilde x_2)+\mathbf{Q}_K(\tilde x_2)\right]  \left[\mathbf{I}-\beta \mathbf{Q}_K\right]^{-1}\mathbf{m},
\end{split}
\end{equation}
For any discount factor that solves \eqref{eq:diffAD}, unique primitive utilities can be found that rationalize the choice data, and these utilities satisfy exclusion restriction \eqref{eq:exclusionrestriction}.\footnote{The argument in Section \ref{s:empiricalcontent}, which establishes a version of  \citeauthor{ecta02:magnacthesmar}'s Proposition 2, implies that the utilities that rationalize the choice data for a given discount factor solve \eqref{eq:u} for $\mathbf{u}_k$. It follows straightforwardly that they satisfy \eqref{eq:exclusionrestriction} whenever \eqref{eq:diffAD} holds.} So, without further assumptions or data, moment condition \eqref{eq:diffAD} contains all the information about the discount factor in the choice data under exclusion restriction \eqref{eq:exclusionrestriction} and can be used directly for its identification and estimation.\footnote{Additional exclusion restrictions (as in Section \ref{s:multiple}) and functional form assumptions on the utility functions may provide further information on the discount factor. After all, the utilities that rationalize the choice data for a discount factor that solves \eqref{eq:diffAD} may not satisfy these additional constraints.}$^\text{,}$\footnote{Similarly, moment condition \eqref{eq:diffMT} contains all the information about the discount factor under \citeauthor{ecta02:magnacthesmar}'s exclusion restriction on current values.}

Unlike the right hand side of (\ref{eq:diffMT}), the right hand side of (\ref{eq:diffAD}) is not linear in $\beta$. Nevertheless, given data on transition and choice probabilities, it is a well-behaved, known function of $\beta$. It is therefore easy to characterize the identified set  ${\cal B}$ of values of $\beta\in[0,1)$ that equate it to the known left hand side of (\ref{eq:diffAD}). 
\begin{theorem}
\label{th:ident}
Suppose that the exclusion restriction in \eqref{eq:exclusionrestriction} holds for some $k\in{\cal D}/\{K\}$, $l \in {\cal D}$, $\tilde x_1 \in {\cal X}$, and $\tilde x_2\in{\cal X}$; with either $k \neq l$,  $\tilde{x}_1\not =\tilde{x}_2$, or both. Moreover, suppose that either the left hand side of (\ref{eq:diffAD}) is nonzero (that is, $p_k(\tilde x_1)/p_K(\tilde x_1)\neq p_l(\tilde x_2)/p_K(\tilde x_2)$)
or a generalization of \citeauthor{ecta02:magnacthesmar}'s rank condition \eqref{eq:rankMT} holds:
\begin{equation}
\label{eq:rank}
\left[\mathbf{Q}_k(\tilde x_1)-\mathbf{Q}_K(\tilde x_1)- \mathbf{Q}_l(\tilde x_2)+\mathbf{Q}_K(\tilde x_2)\right]\mathbf{m}\not = 0.
\end{equation} 
Then, the identified set ${\cal B}$ is a closed discrete subset of $[0,1)$.
\end{theorem}
\begin{proof}
We need to show that, under the stated conditions, ${\cal B}\subseteq[0,1)$ has no limit points in $[0,1)$. First note that $\left[\mathbf{I}-\beta \mathbf{Q}_K\right]^{-1}$ exists for $\beta\in(-1,1)$ and equals 
\begin{equation}
\label{eq:series}
\mathbf{I}+\beta\mathbf{Q}_K+\beta^2 \mathbf{Q}_K^2+\cdots.
\end{equation}
This is trivial for $\beta=0$. If $|\beta|\in(0,1)$, it follows from the facts that $|\beta^{-1}|>1$ and that $\mathbf{Q}_K$ is a Markov transition matrix, with eigenvalues no larger than 1 in absolute value. Consequently, the determinant of $\mathbf{Q}_K-\beta^{-1}\mathbf{I}$ is nonzero, so that $\mathbf{I}-\beta \mathbf{Q}_K$ is invertible and the power series in \eqref{eq:series} converges.

It follows that, for given choice and transition probabilities, the right hand side of \eqref{eq:diffAD} minus its left hand side is a real-valued power series in $\beta$ that converges on $(-1,1)$. Denote the function of $\beta$ this defines with $f:(-1,1)\rightarrow\mathbb{R}$. Corollary 1.2.4 in \citet{birk02:krantzparks} ensures that $f$ is real analytic. 

Denote ${\cal B}^*=\{\beta\in (-1,1)\; |\; f(\beta)=0\}$. Note that ${\cal B}={\cal B}^*\cap[0,1)$.  First, suppose that $f$ has no zeros (${\cal B}^*=\emptyset$). Then, ${\cal B}=\emptyset$ has no limit point in $[0,1)$. 

Finally, suppose that $f$ has at least one zero (${\cal B}^*\neq\emptyset$). Then, $f$ cannot be constant (and thus equal zero) under the stated conditions: If the left hand side of \eqref{eq:diffAD} is nonzero then, because its right hand side equals zero at $\beta=0$, $f(0)$ is nonzero; if rank condition \eqref{eq:rank} holds, then the derivative of the right hand side of \eqref{eq:diffAD} at $\beta=0$, and therefore of $f$ at $0$, is nonzero. Because $f$ is a nonconstant real-analytic function, its zero set ${\cal B}^*$ has no limit point in $(-1,1)$ (\citeauthor{birk02:krantzparks}, Corollary 1.2.7). Because ${\cal B}={\cal B}^*\cap[0,1)$, this implies that ${\cal B}$ has no limit point in $[0,1)$. 
\end{proof}

Under the conditions of Theorem \ref{th:ident}, each $\beta\in[0,1)$ that is consistent with (\ref{eq:diffAD}) is an isolated point in $[0,1)$ and thus locally identified. Note that $\beta=1$ is excluded from the model to ensure convergence of the discounted utility flows. Theorem \ref{th:ident} does not exclude that $1$ is a limit point of the identified set. So, the identified set may contain countably many discount factors near $1$. However, because a closed discrete set is finite on compact subsets, only finitely many discount factors in the identified set lie outside a neighborhood of $1$.
\begin{corollary}
\label{cor:finite}
Under the conditions of Theorem \ref{th:ident}, ${\cal B}\cap[0,1-\epsilon]$ is finite for $0<\epsilon<1$.
\end{corollary}

\noindent In many applications, one may be able to argue against discount factors that are arbitrarily close to $1$. Corollary \ref{cor:finite} shows that, in such applications, it suffices to search for the finite number of discount factors in a compact set $[0,1-\epsilon]$ that solve (\ref{eq:diffAD}), which is computationally easy. 

The right hand side of  \eqref{eq:diffAD} is the log choice probability difference implied by the model with an exclusion restriction across choices $k$ and $l$ and states $\tilde x_1$ and $\tilde x_2$. From the proof of Theorem \ref{th:ident}, we know it equals the discount factor $\beta$, which represents how much the agent cares about the next period, multiplied by the sum of two terms that capture how much relevant variation in next period's expected discounted utility there is for the agent to care about: 
\begin{equation}
\label{eq:firstterm}
\left[\mathbf{Q}_k(\tilde x_1)-\mathbf{Q}_K(\tilde x_1)- \mathbf{Q}_l(\tilde x_2)+\mathbf{Q}_K(\tilde x_2)\right] \mathbf{m}
\end{equation}
 and 
\begin{equation}
\label{eq:secondterm}
\begin{split}
&\left[\mathbf{Q}_k(\tilde x_1)-\mathbf{Q}_K(\tilde x_1)- \mathbf{Q}_l(\tilde x_2)+\mathbf{Q}_K(\tilde x_2)\right] \mathbf{v}_K=\\
&\left[\mathbf{Q}_k(\tilde x_1)-\mathbf{Q}_K(\tilde x_1)- \mathbf{Q}_l(\tilde x_2)+\mathbf{Q}_K(\tilde x_2)\right]  \left[\beta \mathbf{Q}_K+ \beta^2 \mathbf{Q}_K^2+\cdots\right]\mathbf{m}.
\end{split}
\end{equation}
The first term \eqref{eq:firstterm} does not depend on $\beta$. It is nonzero if the generalized rank condition \eqref{eq:rank} holds. It corresponds to the leading, linear term in the right hand side of \eqref{eq:diffAD}, which extends the right hand side of \citeauthor{ecta02:magnacthesmar}'s \eqref{eq:diffMT} to the possibility of comparing across distinct choices $k$ and $l$. 

The next section gives conditions under which the second term \eqref{eq:secondterm} vanishes. Section \ref{s:examples} gives economic examples in which these conditions hold. If they hold, the right hand side of \eqref{eq:diffAD} is linear in $\beta$, so that \eqref{eq:diffAD} uniquely determines $\beta$ under the generalized rank condition \eqref{eq:rank}.

In general, the second term \eqref{eq:secondterm}  does not vanish and depends on $\beta$. Then, the right hand side of  \eqref{eq:diffAD} is not linear in $\beta$, but its derivative at $\beta=0$ still equals the first term \eqref{eq:firstterm}.\footnote{The derivative corresponding to the second term \eqref{eq:secondterm} vanishes because choice $K$ has zero value if the agent is myopic.} Therefore, if the generalized rank condition  \eqref{eq:rank} holds,  this derivative is nonzero and myopic preferences ($\beta=0$) are locally identified.\footnote{Here, $\beta$ is locally identified at some $\beta_0$ if $\beta=\beta_0$ uniquely solves \eqref{eq:diffAD} in a neighborhood of $\beta_0$. Rank condition \eqref{eq:rank} is not {\em necessary} for local identification of $\beta$ at zero; for that, higher order variation of the right hand side of  \eqref{eq:diffAD} in $\beta$ at zero would suffice \citep{ecta83:sargan}.} In economic terms, the rank condition ensures that there is variation in next period's expected discounted utility for a myopic agent to care about, so that only myopic preferences can explain a lack of choice response. In Theorem \ref{th:ident}, the rank condition excludes the trivial case that a zero choice response is observed and the right hand side of \eqref{eq:diffAD} equals zero for all $\beta$.

In the case that a zero choice response is observed, local identification of myopic preferences does not rule out that the data are also consistent with some positive discount factors, as there may be $\beta\in(0,1)$ such that the sum of \eqref{eq:firstterm} and \eqref{eq:secondterm} is zero (that is, there is no variation in next period's expected discounted utility for the agent to care about). These discount factors, if any, can easily be found by searching for the solutions of \eqref{eq:diffAD}. In particular, if the sum of  \eqref{eq:firstterm} and \eqref{eq:secondterm} is nonzero for all $\beta\in(0,1)$, only myopic preferences can explain the lack of choice response.

More generally, rank condition \eqref{eq:rank} does not suffice for point identification of $\beta$. As the next example demonstrates, the same observed choice response may arise from a combination of a low $\beta$ (little care about the next period) and a large absolute sum of \eqref{eq:firstterm} and \eqref{eq:secondterm} (lots of variation in the next period to care about) and from a combination of a high $\beta$ and a small absolute sum of \eqref{eq:firstterm} and \eqref{eq:secondterm}.
\begin{figure}[t]
\caption{Example in Which the Rank Condition Holds, but Identification Fails\label{fig:identificationFails}}
\centering
\medskip
\begin{tikzpicture}
   \begin{axis}[width=\figwidth,height=\figheight,
        name=llherr,
        xlabel={$\beta$},
        axis x line=bottom,
        every outer x axis line/.append style={-,color=gray,line width=0.75pt},
        axis y line=left,
        every outer y axis line/.append style={-,color=gray,line width=0.75pt},
        xtick={0,\identificationFailsPreciseBetaOne,\identificationFailsPreciseBetaTwo},
        xticklabels={$0$,$\identificationFailsBetaOne$,$\identificationFailsBetaTwo$},
        ytick={0,\identificationFailsPreciseLhs},
        yticklabels={$0$,$\identificationFailsLhs$},
        every tick/.append style={line width=0.75pt},
	xmin=0, 
	xmax=1,
	ymin=0,
        scaled ticks=false,
        /pgf/number format/precision=2,
        /pgf/number format/set thousands separator={}]
        \draw[color=black!75,line width=0.75pt,solid] (axis cs:0,\identificationFailsPreciseLhs) -- (axis cs:1,\identificationFailsPreciseLhs);
        \addplot[color=blue!75,smooth,mark=,line width=0.75pt] table[x=beta,y=rhs,col sep=comma]{matlab/data/identificationFails.csv} ;
		\draw[color=blue!75,line width=0.75pt,dotted] (axis cs:\identificationFailsPreciseBetaOne,0) -- (axis cs:\identificationFailsPreciseBetaOne,\identificationFailsPreciseLhs);
         	\draw[color=blue!75,line width=0.75pt,dotted] (axis cs:\identificationFailsPreciseBetaTwo,0) -- (axis cs:\identificationFailsPreciseBetaTwo,\identificationFailsPreciseLhs);
        	\addplot[color=red!75,dashed,mark=,line width=0.75pt] table[x=beta,y=rhsMT,col sep=comma]{matlab/data/identificationFails.csv} ;
        		\draw[color=red!75,line width=0.75pt,dotted] (axis cs:\identificationFailsPreciseBetaMT,0) -- (axis cs:\identificationFailsPreciseBetaMT,\identificationFailsPreciseLhs);
\end{axis}
\end{tikzpicture}
\medskip

\begin{minipage}{\figwidth}
{\scriptsize Note: For $J=\identificationFailsJ$ states, $K=\identificationFailsK$ choices, $k=l=1$, $\tilde{x}_1=x_1$, and $\tilde{x}_2=x_2$, this graph plots the left hand side of  (\ref{eq:diffMT})  and (\ref{eq:diffAD}) (solid black horizontal line) and the right hand sides of (\ref{eq:diffMT}) (dashed red line), and (\ref{eq:diffAD}) (solid blue curve) as functions of $\beta$. The data are $\mathbf{Q}_1(x_1)=\identificationFailsQiRowOne$, $\mathbf{Q}_1(x_2)=\identificationFailsQiRowTwo$,}

{\scriptsize 
\vspace*{-3ex} 
\[\mathbf{Q}_K=\identificationFailsQK\text{, }\mathbf{p}_1=\identificationFailsPOne\text{, and }\mathbf{p}_K=\identificationFailsPK.\] 

Consequently, the left hand side of (\ref{eq:diffMT})  and (\ref{eq:diffAD}) equals $\ln\left(p_1(x_1)/p_K(x_1)\right) - \ln\left(p_1(x_2)/p_K(x_2)\right)= \identificationFailsLhs$. Moreover,  $\mathbf{m}'=\identificationFailsMPrime$ and $\mathbf{Q}_1(x_1)-\mathbf{Q}_K(x_1)- \mathbf{Q}_1(x_2)+\mathbf{Q}_K(x_2)=\identificationFailsDeltaQ$, so that the slope of the dashed red line equals $\left[\mathbf{Q}_1(x_1)-\mathbf{Q}_K(x_1)- \mathbf{Q}_1(x_2)+\mathbf{Q}_K(x_2)\right] \mathbf{m}=\identificationFailsDeltaQM$. A unique value of $\beta$,  $\identificationFailsBetaMT$, solves  (\ref{eq:diffMT}), but two values of  $\beta$ solve (\ref{eq:diffAD}): $\identificationFailsBetaOne$ and $\identificationFailsBetaTwo$.}
\end{minipage}
\end{figure}
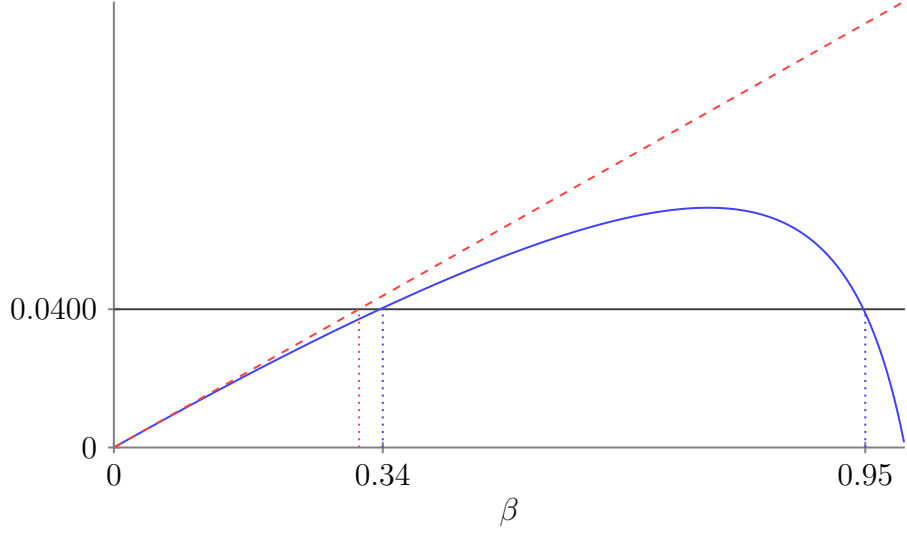

\begin{example}
\label{ex:identificationFails}
Figure \ref{fig:identificationFails} plots the left hand side of \eqref{eq:diffMT} and \eqref{eq:diffAD} (solid black line) and the right hand sides of \eqref{eq:diffMT} (dashed red line) and \eqref{eq:diffAD} (solid blue curve) for a specific example with $K=2$ choices, $k=l=1$, and $J=3$ states. The example's data satisfy the rank condition in \eqref{eq:rank}. Under the current value restriction, there is a unique discount factor that rationalizes the data (the intersection of the black and the dashed red curve). Under the primitive utility restriction, there are two discount factors that rationalize the same data (the intersections of the black and the solid blue curve).
\end{example}

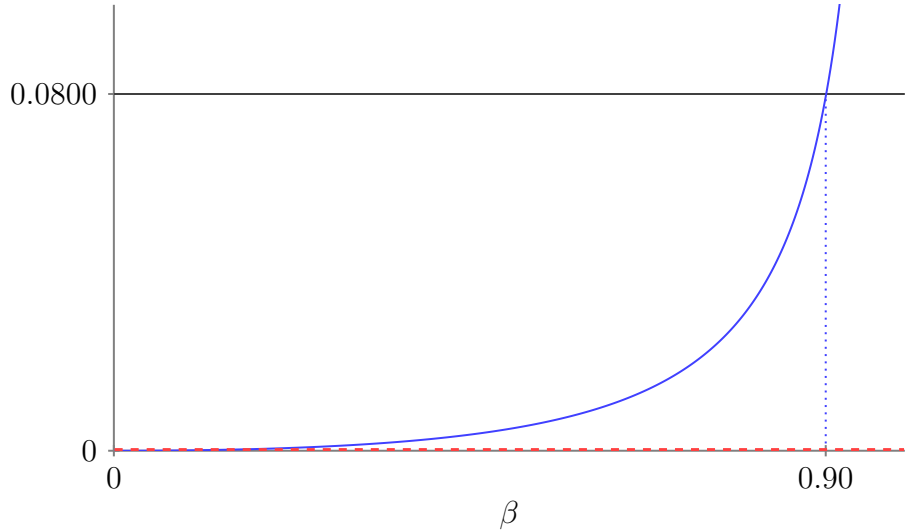
\begin{figure}[t]
\caption{Example in Which the Rank Condition Fails, but the Discount Factor is Identified\label{fig:needNoMT}}
\centering
\medskip
\begin{tikzpicture}
    \begin{axis}[width=\figwidth,height=\figheight,
        name=llherr,
        xlabel={$\beta$},
        axis x line=bottom,
        every outer x axis line/.append style={-,color=gray,line width=0.75pt},
        axis y line=left,
        every outer y axis line/.append style={-,color=gray,line width=0.75pt},
        xtick={0,\needNoMTPreciseBetaOne},
        xticklabels={$0$,$\needNoMTBetaOne$},
        ytick={0,\needNoMTPreciseLhs},
        yticklabels={$0$,$\needNoMTLhs$},
        every tick/.append style={line width=0.75pt},
        	xmin=0, 
	xmax=1,
	ymin=0,
	ymax=0.1,
        scaled ticks=false,
        /pgf/number format/precision=2,
        /pgf/number format/set thousands separator={}]
        \draw[color=black!75,line width=0.75pt,solid] (axis cs:0,\needNoMTPreciseLhs) -- (axis cs:1,\needNoMTPreciseLhs);
        	\draw[color=blue!75,line width=0.75pt,dotted] (axis cs:\needNoMTPreciseBetaOne,0) -- (axis cs:\needNoMTPreciseBetaOne,\needNoMTPreciseLhs);
        \addplot[color=blue!75,smooth,mark=,line width=0.75pt] table[x=beta,y=rhs,col sep=comma]{matlab/data/needNoMT.csv};
        	\addplot[color=red!75,dashed,mark=,line width=2pt] table[x=beta,y=rhsMT,col sep=comma]{matlab/data/needNoMT.csv};
\end{axis}
\end{tikzpicture}
\medskip

\begin{minipage}{\figwidth}
{\scriptsize Note: For $J=\needNoMTJ$ states, $K=\needNoMTK$ choices, $k=l=1$, $\tilde{x}_1=x_1$, and $\tilde{x}_2=x_2$, this graph plots the left hand side of  (\ref{eq:diffMT})  and (\ref{eq:diffAD}) (solid black horizontal line) and the right hand sides of (\ref{eq:diffMT}) (dashed red line) and (\ref{eq:diffAD}) (solid blue curve) as functions of $\beta$. The data are $\mathbf{Q}_1(x_1)=\needNoMTQiRowOne$, $\mathbf{Q}_1(x_2)=\needNoMTQiRowTwo$,}

{\scriptsize 
\vspace*{-3ex} 
\[\mathbf{Q}_K=\needNoMTQK\text{, }\mathbf{p}_1=\needNoMTPOne \text{, and }\mathbf{p}_K=\needNoMTPK.\] 

Consequently, the left hand side of (\ref{eq:diffMT})  and (\ref{eq:diffAD}) equals $\ln\left(p_1(x_1)/p_K(x_1)\right) - \ln\left(p_1(x_2)/p_K(x_2)\right)= \needNoMTLhs$. Moreover,  $\mathbf{m}'=\needNoMTMPrime$ and $\mathbf{Q}_1(x_1)-\mathbf{Q}_K(x_1)- \mathbf{Q}_1(x_2)+\mathbf{Q}_K(x_2)=\needNoMTDeltaQ$, so that the slope of the dashed red line equals $\left[\mathbf{Q}_1(x_1)-\mathbf{Q}_K(x_1)- \mathbf{Q}_1(x_2)+\mathbf{Q}_K(x_2)\right] \mathbf{m}=\needNoMTDeltaQM$. A unique value of $\beta$,  $\needNoMTBetaOne$, solves  (\ref{eq:diffAD}), but (\ref{eq:diffMT}) has no solution.}
\end{minipage}
\end{figure}

Our next example shows that the rank condition in \eqref{eq:rank} is not necessary for point identification either.
\begin{example}
\label{ex:needNoMT}
Figure \ref{fig:needNoMT} presents an example in which \eqref{eq:firstterm} equals zero, so that the right hand side of \eqref{eq:diffMT} and the first (excess surplus) term in the right hand side of \eqref{eq:diffAD} are zero, but the second (value of choice $K$) term in the right hand side of \eqref{eq:diffAD} is positive and increasing with $\beta$. There exists exactly one $\beta\in[0,1)$ that solves \eqref{eq:diffAD}, despite the violation of the rank condition. 

Also note that there is no value of $\beta$ that satisfies \eqref{eq:diffMT}. Even though the data can be rationalized by some specification of the model, they are not consistent with the current value restriction. In other words, this restriction has empirical content.  We return to this point in Section \ref{s:empiricalcontent}. 
\end{example}

\subsection{Finite dependence}
\label{ss:finitedep}

Some of the examples in the next subsection display a variant of \cites{arcidiaconomiller11}  ``finite dependence''. Finite dependence is a property of dynamic discrete choice models that can considerably simplify estimation and is widely used in applications \cite[see][for references]{arcidiaconomillerfinitedependence15}. 

In our context, finite dependence implies that the moment condition is of finite and known polynomial order. This order provides an upper bound on the number of solutions for the discount factor in $\mathbb{R}$, and therefore in $[0,1)$. For example, in the case with $k\neq l=K$,  \eqref{eq:secondterm} reduces to
\begin{equation}
\label{eq:secondtermfindep}
\left[\mathbf{Q}_k(\tilde x_1)-\mathbf{Q}_K(\tilde x_1)\right] \mathbf{v}_K=\left[\mathbf{Q}_k(\tilde x_1)-\mathbf{Q}_K(\tilde x_1)\right]  \left[\beta \mathbf{Q}_K+ \beta^2 \mathbf{Q}_K^2+\cdots\right]\mathbf{m}.
\end{equation}
Suppose that $\mathbf{Q}_k(\tilde x_1)\mathbf{Q}_K^\rho=\mathbf{Q}_K(\tilde x_1)\mathbf{Q}_K^\rho$ for some $\rho\in\{1,2,\ldots\}$. That is, the distribution of the state $\rho+1$ periods from now does not depend on whether the agent chooses $k$ or $K$ now, provided that she follows up in both cases by choosing $K$ in the next $\rho$ periods (independently of whether this is optimal or not). Under this ``single action ($K$) $\rho$-period dependence'' \citep{arcidiaconomillerlongshort} on choices $k$ and $K$ in state $\tilde x_1$, $\mathbf{Q}_k(\tilde x_1)\mathbf{Q}_K^r=\mathbf{Q}_K(\tilde x_1)\mathbf{Q}_K^r$ for all $r\in\{\rho,\rho+1,\ldots\}$.\footnote{Throughout, we focus on this special case of \cites{arcidiaconomiller11} finite dependence, which turns out to be particularly powerful in our specific context.} Now assume that Theorem \ref{th:ident}'s conditions hold. If $\rho=1$, the right hand side of \eqref{eq:secondtermfindep} equals zero, the current value 
\[
U_k(\tilde x_1)=u^*_k(\tilde x_1) + \beta\left[\mathbf{Q}_k(\tilde x_1)-\mathbf{Q}_K(\tilde x_1)\right]\mathbf{v}_K=u^*_k(\tilde x_1),
\]  
the right hand side of \eqref{eq:diffAD}  is linear in $\beta$, and $\beta$ is point identified. If instead $\rho\geq 2$, then the right hand side of \eqref{eq:secondtermfindep}  equals
\begin{equation*}
\left[\mathbf{Q}_k(\tilde x_1)-\mathbf{Q}_K(\tilde x_1)\right]  \left[\beta \mathbf{Q}_K+ \cdots +\beta^{\rho-1} \mathbf{Q}_K^{\rho-1}\right]\mathbf{m},
\end{equation*}
the right hand side of \eqref{eq:diffAD} is a $\rho$-th order polynomial in $\beta$, and the identified set ${\cal B}$ holds no more than $\rho$ discount factors. This example straightforwardly extends to the general exclusion restriction in \eqref{eq:exclusionrestriction}, which we state without further proof.
\begin{theorem}\label{th:finitedependence}
Suppose that the conditions of Theorem \ref{th:ident} hold and that $\{\mathbf{Q}_k; k\in{\cal D}\}$  satisfies single action ($K$) $\rho$-period dependence on choices $k$ and $K$ in state ${\tilde x}_1$,
\[  
\mathbf{Q}_k(\tilde x_1)\mathbf{Q}_K^\rho=\mathbf{Q}_K(\tilde x_1)\mathbf{Q}_K^\rho,
\]
and single action ($K$) $\rho$-period dependence on choices $l$ and $K$ in state ${\tilde x}_2$,
\[
\mathbf{Q}_l(\tilde x_2)\mathbf{Q}_K^\rho=\mathbf{Q}_K(\tilde x_2)\mathbf{Q}_K^\rho,
\]
for some $\rho\in\{1,2,\ldots\}$. Then there are no more than $\rho$ points in the identified set ${\cal B}$. 
\end{theorem}

Theorem \ref{th:finitedependence} applies finite dependence to cancel differences in expected discounted utilities across pairs of choices twice, once for each of the two states that appear in the exclusion restriction. In the special case that the exclusion restriction concerns a comparison across states $\tilde x_1$ and $\tilde x_2$ for a given choice $k=l$, the right hand side of \eqref{eq:secondterm} reduces to 
\begin{equation}
\label{eq:secondtermStates}
\left[\mathbf{Q}_k(\tilde x_1)-\mathbf{Q}_K(\tilde x_1)- \mathbf{Q}_k(\tilde x_2)+\mathbf{Q}_K(\tilde x_2)\right]  \left[\beta \mathbf{Q}_K+ \beta^2 \mathbf{Q}_K^2+\cdots\right]\mathbf{m}.
\end{equation}
By Theorem \ref{th:finitedependence}, single action ($K$) $\rho$-period dependence on choices $k$ and $K$ in states ${\tilde x}_1$  and ${\tilde x}_2$  implies that the identified set contains at most $\rho$ discount factors. If $\rho=1$, then both $U_k(\tilde x_1)=u^*(\tilde x_1)$ and $U_k(\tilde x_2)=u^*(\tilde x_2)$, \eqref{eq:secondtermStates} equals $0$, and the discount factor is point identified.

In this case with $k=l$, the consequent of Theorem \ref{th:finitedependence} would also hold if, alternatively,
\[
\mathbf{Q}_k(\tilde x_1)\mathbf{Q}_K^\rho=\mathbf{Q}_k(\tilde x_2)\mathbf{Q}_K^\rho ~~~\text{and}~~~ \mathbf{Q}_K(\tilde x_1)\mathbf{Q}_K^\rho=\mathbf{Q}_K(\tilde x_2)\mathbf{Q}_K^\rho,
\]
for some $\rho\in\{1,2,\ldots\}$. This is a form of single action ($K$) $\rho$-period dependence on the initial state (instead of the initial choice) under, respectively, choices $k$ and $K$. Under one-period dependence on initial states $\tilde x_1$ and $\tilde x_2$, current values do not necessarily reduce to primitive utilities, but it is still true that $U_k(\tilde x_1)-U_k(\tilde x_2)=u^*_k(\tilde x_1)-u^*_k(\tilde x_2)$, \eqref{eq:secondtermStates} equals $0$, and the discount factor is point identified.

\subsection{Examples}\label{s:examples}
Theorem \ref{th:ident} shows that the identified set of discount factors is discrete and, away from one, finite, but does not establish point identification. Indeed, Example \ref{ex:identificationFails} demonstrated that point identification may fail, even if rank condition \eqref{eq:rank} holds. 

Our first two examples below (Examples \ref{ex:medicare} and \ref{ex:loyaltyprogram}) illustrate applications from the literature in which the exclusion restriction is plausibly met, but the discount factor is not necessarily point identified.  We then give two examples (Examples \ref{ex:zurcher} and \ref{ex:dixit}) of optimal stopping problems with single action one-period dependence, which are point identified by Theorem \ref{th:finitedependence}. Finally, Example \ref{ex:experience} demonstrates that single action one-period dependence is not necessary for point identification. It is a labor supply model that does not exhibit such one-period dependence, but in which monotonicity of the moment condition in the discount factor gives point identification.  

A number of empirical studies of demand for health care insured under Medicare Part D base their identification on nonlinearities in the price schedules. Our first example describes an empirical strategy from this literature in which the primitive utility exclusion restriction seems plausibly met.
\begin{example}\label{ex:medicare}
As part of an informal identification argument, \cite{einavetal15} observed that changes in purchase behaviour around kinks in the price of insurance are informative on time preferences. In the data, insurees pay 25\% of additional expenditures out of pocket as long as their total yearly expenditures range between \$275 and \$2510, but contribute 100\% to all expenditures between \$2510 and \$5726. A myopic insuree would change her spending only after her total spending hits $\$2510$ and her out-of-pocket contributions increase. In contrast, a forward looking insuree who is close to the kink late in the year would limit her spending before hitting the increase in contributions. Changes in the propensity to spend towards the end of the year for those close to the kink are therefore taken to be informative on the discount factor. 

This argument  can be represented as an exclusion restriction. Let $x$ be the yearly expenditure, a state controlled by the choice of filling prescriptions. The utility $u_k(x)$ of a particular drug purchase $k$ is assumed constant for two expenditure levels $\tilde x_1 < \tilde x_2$ in $[275, 2510)$. Along with variation in expected future expenses because of the kinked price schedule, the exclusion restriction gives set identification by Theorem \ref{th:ident}.\footnote{\cite{gowrisankarantown15} used a similar argument to identify salience and myopia in a dynamic discrete choice model with parametric utility applied to Medicare Part D data.} 
\end{example}
The next example is from \cite{rossi18} which studied the effect of reward programs on gasoline sales using a dynamic discrete choice model.
\begin{example}\label{ex:loyaltyprogram}
In each period, a consumer can choose to buy gasoline ($k = 1$), or not ($k = K$). Consumers accumulate reward points $x$ by registering their gasoline purchases.  The accumulated points can be traded against nonpecuniary rewards at various point thresholds $\underline x$.  \citeauthor{rossi18} observed that the purchase frequency is accelerating in the accumulated points. By assuming that the current period payoff of a gasoline purchase at any price $\tilde y$ is independent of the accumulated points, i.e. that $u_k(\tilde y,\tilde x_1) = u_k(\tilde y, \tilde x_2)$ for all $\tilde y$ and $\tilde x_1, \tilde x_2 \in [0,\underline x)$, the purchase acceleration is informative on the discount factor. The closer the consumer is to qualify for a given reward, the less the future reward is discounted. This makes a current purchase more attractive and predicts a purchase frequency that is increasing in the reward points.
\end{example}

We next turn to optimal stopping problems. The first example is the bus engine replacement problem of \citet{ecta87:rust}. Though the plausibility of the exclusion restriction is questionable in this particular application, it illustrates how one-period dependence gives point identification in a well-known application of an optimal stopping model.
\begin{example}\label{ex:zurcher}
\citet{ecta87:rust} studied Harold Zurcher's management of a fleet of (independent) buses. In each period, Zurcher can either operate a bus as usual ($d=1$) or renew its engine ($d=K=2$). The payoff from operating the bus as usual depend on its mileage $x$ since last renewal, which both Zurcher and \citeauthor{ecta87:rust} observe, and an additive and independent shock. Renewal incurs a cost that is independent of mileage and resets mileage to $x_1=0$:
\[
\mathbf{Q}_K=\left[\begin{array}{cccc}
		1&\cdots&0&0\\
		\vdots&&\vdots&\vdots\\
		1&\cdots&0&0
	\end{array}\right].
\]
In terms of Section \ref{ss:finitedep}'s finite dependence, mileage is single action ($K$) one-period dependent on both initial mileage and the initial renewal choice. Consequently, Zurcher's expected discounted payoffs from renewal do not depend on mileage. In particular, with our normalization $\mathbf{u}_K=\mathbf{0}$,  $v^*_K(\tilde{x})=\beta \left(m(x_1) + v^*_K(x_1)\right)$ for all $\tilde{x}\in{\cal X}$. Since $v^*_K(\tilde{x})$ does not vary with $\tilde x$, $[\mathbf{Q}_1-\mathbf{Q}_K]\mathbf{v}_K=\mathbf{0}$, and $U_1(\tilde x) = u_1^*(\tilde x)$. Therefore, if we assume $u^*_1(\tilde x_1)=u^*_1(\tilde x_2)$, which may be questionable in this application,  \citeauthor{ecta02:magnacthesmar}'s exclusion restriction holds and its identification result applies.\footnote{Since mileage is the only observed state variable in this application, $u^*_1(\tilde x_1)=u^*_1(\tilde x_2)$ requires that the current payoffs from operating a bus are the same at $\tilde x_1$ and $\tilde x_2$ miles, for example because $\tilde x_1$ and $\tilde x_2$ lie on a known flat segment of Harold Zurcher's cost curve.} Its rank condition  (\ref{eq:rankMT}) simplifies to
\[
\left[\mathbf{Q}_1(\tilde x_1)- \mathbf{Q}_1(\tilde x_2)\right]\mathbf{m}\neq 0.
\]
That is, it simply requires that the expected next period's excess surplus differs between states $\tilde x_1$ and $\tilde x_2$ under continued operation of the bus (choice 1). 
\end{example}

Example \ref{ex:zurcher}'s analysis of optimal renewal extends to optimal stopping problems in which stopping ends the decision problem. For example, in \cites{ecta92:hopenhayn} model of firm dynamics with free entry, active firms solve optimal stopping problems in which they value exit $K$ at $\mathbf{v}_K=\mathbf{0}$. As in Example \ref{ex:zurcher}, the fact that $v^*_K(\tilde x)$ is constant in $\tilde x$ ensures that the expectation contrast $[\mathbf{Q}_1-\mathbf{Q}_K]\mathbf{v}_K=\mathbf{0}$, so that $U_1(\tilde x) = u_1^*(\tilde x)$. 

Of course, $[\mathbf{Q}_1-\mathbf{Q}_K]\mathbf{v}_K$ may equal zero even if $v^*_K(\tilde x)$ varies with $\tilde x$, in particular if the state is single action ($K$) one-period dependent on choices $1$ and $K$.

\begin{example}
\label{ex:dixit}
Consider a discrete time, econometric implementation of \cites{jpe89:dixit} model of firm entry and exit.\footnote{\cite{ddc15:abbringklein} presented this example's model with state independent entry costs, code for its estimation, and exercises that can be used in teaching dynamic discrete choice models.}  In each period, a firm chooses to either serve the market ($d=1$) or not ($d=K=2$). Its payoffs from serving the market depend on $x=(y,d_{-1})$, where $y$ is a profit shifter that follows an exogenous Markov process (that is, $y$ may affect choices but is not controlled by them) and $d_{-1}$ is the firm's choice in the previous period. The entry costs in profit state $\tilde y$ equal the difference between an incumbent's profit from serving the market and a new entrant's profit from doing so, $u^*_1(\tilde y,1)-u^*_1(\tilde y,K)$, which we assume to be nonnegative. As before, we set $\mathbf{u}_K=\mathbf{0}$, so that the exit costs $u^*_K(\tilde y,K)-u^*_K(\tilde y,1)$ are zero. 

The firm's value $v^*_K(y',k)$ from choosing inactivity ($K$) next period after choosing $d=k$ now may vary with next period's profit state $y'$, because the firm will have the option to reenter the market and this option's value may depend on $y'$. However, because exit costs are zero, this value does not depend on the current choice $k$: $v^*_K(y',1)=v^*_K(y',K)$. Moreover, by the assumption that $y$ follows an exogenous Markov process, the distribution of $y'$ given $(y,d_{-1},d=k)$ is independent of the current choice $k$ and the past choice $d_{-1}$, so that
\begin{equation}
\label{eq:QvDixit}
\mathbf{Q}_1(\tilde x)\mathbf{v}_K=\mathbb{E}\left[v^*_K(y',1) \;\vline\; y=\tilde y\right]
=\mathbb{E}\left[v^*_K(y',K) \;\vline\; y=\tilde y\right]=\mathbf{Q}_K(\tilde x)\mathbf{v}_K
\end{equation}

\noindent  for all $\tilde x=(\tilde y,\tilde d_{-1})\in{\cal X}$.  Consequently, as in Example \ref{ex:zurcher}, $[\mathbf{Q}_1(\tilde x)-\mathbf{Q}_K(\tilde x)]\mathbf{v}_K=0$ and $U_1(\tilde x) = u_1^*(\tilde x)$. Note that, in this case, the state $x=(y,d_{-1})$ is single action ($K$) one-period dependent on choices $1$ and $K$ (but generally not on initial states). 

An exclusion restriction $u_1^*(\tilde x_1)=u_1^*(\tilde x_2)$ implies \eqref{eq:diffMT} and, under rank condition \eqref{eq:rankMT},  point identification of $\beta$.  Because $y$ evolves independently of current and past choices, 
\begin{equation}
\label{eq:QmDixit}
\mathbf{Q}_k(\tilde x)\mathbf{m}=\mathbb{E}\left[m(y',k) \;\vline\; y=\tilde y\right]. 
\end{equation}

\noindent Thus, the rank condition is equivalent to
\begin{equation}
\label{eq:rankDixit}
\mathbb{E}\left[m(y',1)-m(y',K)\;\vline\; y=\tilde y_1\right]\not=\mathbb{E}\left[m(y',1)-m(y',K)\;\vline\; y=\tilde y_2\right].
\end{equation}

\noindent It immediately follows that identification requires that $\tilde y_1\not =\tilde y_2$ in this case. A difference in lagged choices alone would not suffice, because these do not help predict next period's profit state $y'$ given the current profit state $y$ and choice $d=k$ nor directly affect next period's excess surplus.

Moreover, identification fails if entry costs are zero; that is, if $u_1^*(\tilde y,1)=u_1^*(\tilde y,K)$. In this case, payoffs do not depend on past choices and, more specifically, $m(y',1)=m(y',K)$. Intuitively, without entry and exit costs, firms can ignore past and future when deciding on entry and exit and simply maximize the current profits in each period. Consequently, their entry and exit choices carry no information on their discount factor. As an aside, note that the entry costs are directly identified from 
\[ 
\ln\left(p_1(\tilde x_1)/p_K(\tilde x_1)\right) - \ln\left(p_1(\tilde x_2)/p_K(\tilde x_2)\right)=
u_1^*(\tilde y,1)- u_1^*(\tilde y,K)
\]

\noindent for $\tilde x_1=(\tilde y,1)$ and $\tilde x_2=(\tilde y,K)$. Intuitively, for given profit state $\tilde y$, lagged choices only affect current payoffs through the entry costs and have no effect on expected future payoffs, as is clear from \eqref{eq:QvDixit} and \eqref{eq:QmDixit}. 

Finally, if both $\tilde y_1\not =\tilde y_2$ and entry costs are strictly positive, \eqref{eq:rankDixit} will generally be satisfied. In specific applications, we can verify \eqref{eq:rankDixit} using that both the distribution of $y'$ conditional on $y$ and $m(y',k)=-\ln\left(p_K(y',k)\right)$ can directly be estimated from choice and profit state transition data.

As in Zurcher's problem, profit states are typically ordered, so that an exclusion restriction like $u_1^*(\tilde x_1)=u_1^*(\tilde x_2)$ may be justified as a local shape restriction on the firm's utility function. Alternatively, because the firm's utility is a cardinal payoff, we may be able to exploit that $u^*_1(\tilde x)$ is known in some state $\tilde x$. For example, if $u_1^*(\tilde x)=0$, then \eqref{eq:diffAD} holds with $k=1$, $l=K$, and $\tilde x_1=\tilde x_2=\tilde x=(\tilde y,\tilde d_{-1})$ and reduces to
\[
\ln\left(p_1(\tilde x)\right)-\ln\left(p_K(\tilde x)\right)=\beta\mathbb{E}\left[m(y',1)-m(y',K)\;\vline\; y=\tilde y\right],
\]

\noindent so that $\beta$ is identified if $\mathbb{E}\left[m(y',1)-m(y',K)\;\vline\; y=\tilde y\right]\not=0$. This rank condition is generally satisfied if entry costs are positive.
\end{example}

In Examples \ref{ex:zurcher} and \ref{ex:dixit}, the rank condition ensures that the shift in expected surplus contrasts that multiplies $\beta$ in the right hand side of \eqref{eq:diffAD} is nonzero. Because these examples satisfy one-period dependence, this shift does not depend on $\beta$ itself, and this suffices for point identification.
More generally,  even if the state is not one-period dependent, strict monotonicity of the right hand side of (\ref{eq:diffAD}), as in Example \ref{ex:needNoMT}, suffices for point identification (that is, ensures that a solution is unique if it exists). It is easy to derive conditions that imply such strict monotonicity, and thus point identification, and that do not involve $\beta$. Without loss of generality--- we can freely interchange states $\tilde x_1$ and $\tilde x_2$ and switch choices $k$ and $l$--- we focus on conditions under which it is strictly increasing or, equivalently, its derivative with respect to $\beta$ is positive:\footnote{Denoting $\Delta^2\mathbf{Q}\equiv \mathbf{Q}_k(\tilde x_1)-\mathbf{Q}_K(\tilde x_1)- \mathbf{Q}_l(\tilde x_2)+\mathbf{Q}_K(\tilde x_2)$, we have that
\begin{align*}
\frac{\partial}{\partial\beta}  \beta\Delta^2\mathbf{Q}\left[\mathbf{I}-\beta \mathbf{Q}_K\right]^{-1}\mathbf{m}
 & = \Delta^2 \mathbf{Q} [\mathbf{I} - \beta \mathbf{Q}_K]^{-1}\mathbf{m} + \beta \Delta^2 \mathbf{Q} \frac{\partial [\mathbf{I} - \beta \mathbf{Q}_K]^{-1}}{\partial \beta} \mathbf{m} \\
& = \Delta^2 \mathbf{Q} [\mathbf{I} - \beta \mathbf{Q}_K]^{-1}\mathbf{m} + \beta \Delta^2 \mathbf{Q} [\mathbf{I} - \beta \mathbf{Q}_K]^{-1}\mathbf{Q}_K [\mathbf{I} - \beta \mathbf{Q}_K]^{-1}\mathbf{m} \\
& =  \Delta^2 \mathbf{Q}[\mathbf{I} + (\mathbf{I} - \beta \mathbf{Q}_K)^{-1}\beta\mathbf{Q}_K][\mathbf{I} - \beta \mathbf{Q}_K]^{-1}\mathbf{m} 
 = \Delta^2 \mathbf{Q}[\mathbf{I}- \beta \mathbf{Q}_K]^{-2}\mathbf{m}.
\end{align*}
}
\[
\left[\mathbf{Q}_k(\tilde x_1)-\mathbf{Q}_K(\tilde  x_1)- \mathbf{Q}_l(\tilde x_2)+\mathbf{Q}_K(\tilde x_2)\right]  \left[\mathbf{I}-\beta \mathbf{Q}_K\right]^{-2}\mathbf{m}>0.
\]
For this, it suffices that
\begin{equation}\label{eq:rankmonotone}
\left[\mathbf{Q}_k(\tilde x_1)-\mathbf{Q}_K(\tilde x_1)- \mathbf{Q}_l(\tilde x_2)+\mathbf{Q}_K(\tilde x_2)\right] \mathbf{Q}_K^r\mathbf{m}\geq 0\text{ for all }r\in\{0,1,2,\ldots\},
\end{equation}
with the inequality strict for at least one $r$. Like \citeauthor{ecta02:magnacthesmar}'s rank condition (\ref{eq:rankMT}), these conditions do not depend on $\beta$. It is easy to verify that they hold in Example \ref{ex:needNoMT} (which is specified in the Note to Figure \ref{fig:needNoMT}).

\begin{figure}[t]
\caption{Example of a Dynamic Labor Supply Model that Gives a Monotone Moment Condition\label{fig:experience}}
\centering
\medskip
\begin{tikzpicture}
   \begin{axis}[width=\figwidth,height=\figheight,
        name=llherr,
        xlabel={$\beta$},
        axis x line=bottom,
        every outer x axis line/.append style={-,color=gray,line width=0.75pt},
        axis y line=left,
        every outer y axis line/.append style={-,color=gray,line width=0.75pt},
        xtick={0,\experiencePreciseBetaOne},
        xticklabels={$0$,$\experienceBetaOne$},
        ytick={0,\experiencePreciseLhs},
        yticklabels={$0$,$\experienceLhs$},
        every tick/.append style={line width=0.75pt},
	xmin=0, 
	xmax=1,
	ymin=0,
	ymax=1,
        scaled ticks=false,
        /pgf/number format/precision=2,
        /pgf/number format/set thousands separator={}]
        \draw[color=black!75,line width=0.75pt,solid] (axis cs:0,\experiencePreciseLhs) -- (axis cs:1,\experiencePreciseLhs);
        	\draw[color=blue!75,line width=0.75pt,dotted] (axis cs:\experiencePreciseBetaOne,0) -- (axis cs:\experiencePreciseBetaOne,\experiencePreciseLhs);
        \addplot[color=blue!75,smooth,mark=,line width=0.75pt] table[x=beta,y=rhs,col sep=comma]{matlab/data/experience.csv};
        	\addplot[color=red!75,dashed,mark=,line width=0.75pt] table[x=beta,y=rhsMT,col sep=comma]{matlab/data/experience.csv};
\end{axis}
\end{tikzpicture}
\medskip

\begin{minipage}{\figwidth}
{\scriptsize Note: For $J=\experienceJ$ states, $K=\experienceK$ choices, $k=l=1$, $\tilde{x}_1=x_2$, and $\tilde{x}_2=x_1$, this graph plots the left hand side of  (\ref{eq:diffMT})  and (\ref{eq:diffAD}) (solid black horizontal line) and the right hand sides of (\ref{eq:diffMT}) (dashed red line) and (\ref{eq:diffAD}) (solid blue curve) as functions of $\beta$ (we switched the roles of $x_1$ and $x_2$ to ensure a positive choice response and visually line up this example with the others). The data are generated from Example \ref{ex:experience}'s stylized dynamic labor supply model, which gives $\mathbf{Q}_1(x_2)=\experienceQiRowOne$, $\mathbf{Q}_1(x_1)=\experienceQiRowTwo$,}

{\scriptsize 
\vspace*{-3ex} 
\[\mathbf{Q}_K=\experienceQK\text{, }\mathbf{p}_1=\experiencePOne \text{, and }\mathbf{p}_K=\experiencePK.\] 

Consequently, the left hand side of (\ref{eq:diffMT})  and (\ref{eq:diffAD}) equals $\ln\left(p_1(x_2)/p_K(x_2)\right) - \ln\left(p_1(x_1)/p_K(x_1)\right)= \experienceLhs$. Moreover,  $\mathbf{m}'=\experienceMPrime$ and $\mathbf{Q}_1(x_2)-\mathbf{Q}_K(x_2)- \mathbf{Q}_1(x_1)+\mathbf{Q}_K(x_1)=\experienceDeltaQ$, so that the slope of the dashed red line equals $\left[\mathbf{Q}_1(x_2)-\mathbf{Q}_K(x_2)- \mathbf{Q}_1(x_1)+\mathbf{Q}_K(x_1)\right] \mathbf{m}=\experienceDeltaQM$. A unique value of $\beta$,  $\experienceBetaOne$, solves  (\ref{eq:diffAD}), but (\ref{eq:diffMT}) has no solution.}
\end{minipage}
\end{figure}
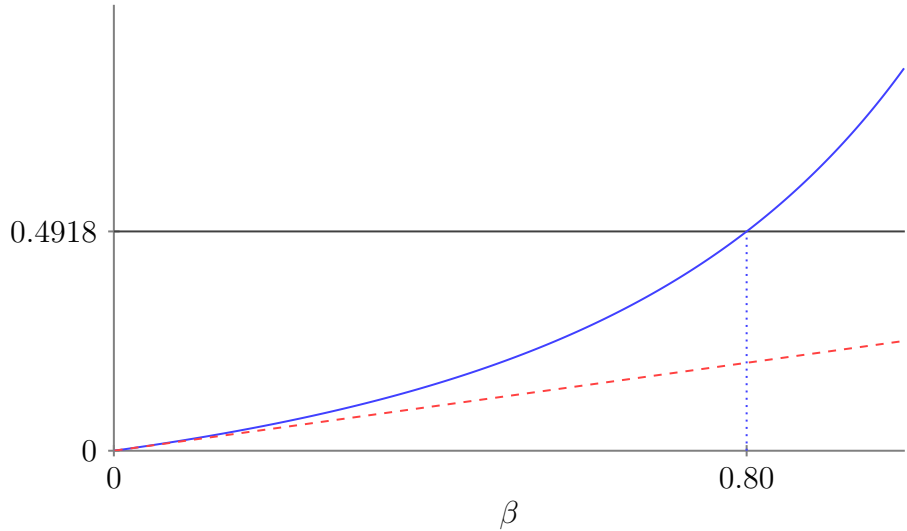

The final example relies on a type of payoff monotonicity that is common in models with ordered states. 
\begin{example}\label{ex:experience}
In \citeauthor{res89:ecksteinwolpin}'s dynamic model of female labor force participation, women work both to directly earn wages and to invest in work experience that pays off later. Consider a highly stylized and stationary variant of this model. Each period, a woman either works ($d=1$) or shirks ($d=2=K$). Work experience takes three levels, ``novice'' ($x_1$), ``learning'' ($x_2$), and ``seasoned'' ($x_3$). If a woman works and is not yet seasoned, her experience increases one level with probability $\experienceStructurelearnRate$ and stays the same with the complementary probability. If instead she shirks, and is not a novice, she falls back one level of experience with probability $\experienceStructuredepRate$ and keeps her experience otherwise. Work gives utility $u_1(x_1)=u_1(x_2)=\experienceStructureuOne$ if novice or learning and $u_1(x_3)=\experienceStructureuThree$ if seasoned. Women maximize their flow of expected utility, discounted with a factor $\experienceStructurebeta$.  

Figure \ref{fig:experience} gives the data implied by this example and plots the moment condition corresponding to the constraint that $u_1(x_1)=u_1(x_2)$. This constraint implies that novices and learning workers earn the same current utility. Nevertheless, work is more attractive to a learning woman, because she has a good shot at earning the higher wage for seasoned workers next period if she works now; moreover, unlike a novice, she may lose experience if she shirks. Seasoned workers, despite the fact that they cannot further increase their experience, are sufficiently motivated by the higher earnings and  the risk that their human capital depreciates to work even more. Consequently, $p_K(x_1)>p_K(x_2)>p_K(x_3)$, so that $m(x_1)<m(x_2)<m(x_3)$. More generally, because  $\mathbf{Q}_K$ is increasing,  the expected excess surplus after $r$ rounds of shirking and human capital depreciation, $\mathbf{Q}_K^r\mathbf{m}$, is increasing in initial experience. 

In this example, the dependence of (the distribution of) a worker's experience on initial choices and experience levels does not disappear in a finite number of periods of, e.g., shirking.\footnote{In a similar context, \citet{res98:altugmiller} impose such finite dependence by assuming that wages and the utility cost from work only depend on a finite employment history. Our example would display single action ($K$) one-period dependence on initial choices in state $x_1$ and two-period dependence in state $x_2$ if shirking women would for sure see their experience drop by one level. Note that this would still not suffice to reduce the moment condition to \citeauthor{ecta02:magnacthesmar}'s linear moment condition.}  In particular, experience is not single action ($K$) one-period dependent on initial choices in states $x_1$ and $x_2$ and \citeauthor{ecta02:magnacthesmar}'s current value restriction and linear moment condition do not hold. Nevertheless, this example's monotonicity ensures that the discount factor is point identified. Because working, compared to shirking, affects the experience of a learning worker more than that of a novice with nothing to lose ($[\mathbf{Q}_1(x_2)-\mathbf{Q}_K(x_2)- \mathbf{Q}_1(x_1)+\mathbf{Q}_K(x_1)]=$ \experienceInlineDeltaQ) and $\mathbf{Q}_K^r\mathbf{m}$ is increasing for all $r$, \eqref{eq:rankmonotone} holds. Consequently, the moment condition is monotone in $\beta$ and has only one solution, $\experienceStructurebeta$.
\end{example}

\subsection{Extension to nonstationary models}
\label{ss:finitehorizon}

Our analysis extends to nonstationary models, such as that in \citeauthor{jpe97:keanewolpin},  with minor modifications. In fact, nonstationary models offer useful identification strategies that are not available for stationary models. Unlike in stationary models, an assumption of stationary utilities has identifying power in nonstationary models. A common version of this argument is that the utilities can be identified in the last period, say $T$, so that the discount factor is subsequently identified in the next to last period \citep[e.g.][]{yaoetal12}. This argument assumes stationary utilities, which can be cast as an exclusion restriction on time as a state variable, i.e. $u_{i,T-1}(\tilde x) = u_{i, T}(\tilde x)$, where time shifts the continuation values without shifting the primitive utilities. 

\cite{qme2016:bajarietal} used the assumption of stationary utilities to formally establish identification in a finite-horizon optimal stopping model. Theorem \ref{th:finitehorizon} below extends \citeauthor{qme2016:bajarietal}'s result beyond optimal stopping problems and also allows for identification of models with nonstationary utilities.\footnote{\citeauthor{yaoetal12} showed identification of the discount factor in a dynamic model with continuous controls under the assumption of stationary utilities and conjectured a similar result for discrete controls. Theorem \ref{th:finitehorizon} proves its conjecture.} 

Denote time by $t\in\{1,2,\ldots,T\}$, with terminal period $T<\infty$, and index $u^*_{k,t}$,  $\mathbf{u}_{k,t}$, $\mathbf{m}_t$, and $\mathbf{v}_{k,t}$  by time. For ease of exposition, we maintain the assumption of stationary Markov transition matrices $\mathbf{Q}_k$, but the results extend to nonstationary distributions. The choice-$k$ specific values now satisfy
\begin{align}
\label{eq:vkt}
	\mathbf{v}_{k,t}& = \mathbf{u}_{k,t} + \beta \mathbf{Q}_k \left[\mathbf{m}_{t+1}+\mathbf{v}_{K,t+1}\right]
\end{align}
for $t=1,\ldots,T-1$; with terminal condition $\mathbf{v}_{k,T}=\mathbf{u}_{k,T}$. With the normalization $\mathbf{u}_{K,t}=\mathbf{0}$ for all $t$, this gives
\begin{align}\label{eq:FHreducedform}
\ln \left(p_{k,t}(\tilde{x})\right)-\ln \left(p_{K,t}(\tilde{x})\right) & = 
u^*_{k,t}(\tilde x) + \beta[\mathbf{Q}_k(\tilde x)- \mathbf{Q}_K(\tilde x)][\mathbf{m}_{t+1} + \mathbf{v}_{K,t+1}] 
\end{align}
for all $k\in {\cal D}\backslash\{K\}$ and $\tilde x \in {\cal X}$. Finally, using \eqref{eq:vkt} and the normalization $\mathbf{u}_{K,t}=\mathbf{0}$ for all $t$, we can write the value of the reference choice $K$ as
\begin{align} 
\label{eq:vcapKt}
\mathbf{v}_{K,t} = \sum^{T}_{\tau = t+1} (\beta \mathbf{Q}_K)^{\tau- t}\mathbf{m}_{\tau},
\end{align}

\noindent where we use the convention that $\sum_{\tau=T+1}^T\cdot=0$ (so that indeed $\mathbf{v}_{K,T}=\mathbf{u}_{K,T}=0$). 

\begin{theorem}\label{th:finitehorizon}
Suppose that 
\begin{align}\label{eq:FHexclusionrestriction}
u^*_{k,t}(\tilde x_1) & = u^*_{l,t'}(\tilde x_2)
\end{align}
for $k\in{\cal D}/\{K\}$, $l \in {\cal D}$, $\tilde x_1 \in {\cal X}$, $\tilde x_2\in{\cal X}$, $1\leq t'<T$, and $t'\leq t\leq T$; with either $k \neq l$, or $\tilde x_1 \neq \tilde x_2$, or $t' < t$, or a combination of the three.  If either 
$p_{k,t}(\tilde x_1)/p_{K,t}(\tilde x_1) \not= p_{l,t'}(\tilde{x}_2)/p_{K,t'}(\tilde x_2)$ or  
\begin{align}\label{eq:FHregularitycondition}
[\mathbf{Q}_k(\tilde x_1)- \mathbf{Q}_K(\tilde x_1)]\mathbf{m}_{t+1}  -
[\mathbf{Q}_l(\tilde x_2) -  \mathbf{Q}_K(\tilde x_2)]\mathbf{m}_{t'+1} \neq 0,
\end{align}
then there are no more than $T-t'$ points in the identified set. 
\end{theorem}
\begin{proof}
Differencing \eqref{eq:FHreducedform} corresponding to \eqref{eq:FHexclusionrestriction} and substituting in \eqref{eq:vcapKt} gives 
\begin{align}\label{eq:timehomogeneity}
\ln \left(p_{k,t}(\tilde x_1)/p_{K,t}(\tilde x_1)\right) &- \ln \left(p_{l,t'}(\tilde{x}_2)/p_{K,t'}(\tilde x_2)\right)  = \\
\beta \Bigg(&\left[\mathbf{Q}_k(\tilde x_1)- \mathbf{Q}_K(\tilde x_1)\right]\Big[\sum^{T}_{\tau = t+1} (\beta \mathbf{Q}_K)^{\tau- t-1}\mathbf{m}_{\tau} \Big]  - \nonumber\\
&\left[\mathbf{Q}_l(\tilde x_2) -  \mathbf{Q}_K(\tilde x_2)\right] \Big[ \sum^{T}_{\tau = t'+1} (\beta \mathbf{Q}_K)^{\tau- t'-1}\mathbf{m}_{\tau}\Big]\Bigg). \nonumber
\end{align}
For given choice and transition probabilities, the right hand side of \eqref{eq:timehomogeneity} minus its left hand side is a polynomial of order $T-t'$ in $\beta$. If this polynomial is nonconstant, then by the fundamental theorem of algebra, it has has up to $T-t'$ real roots, which is an upper bound on the number of points in the identified set. To show that  \eqref{eq:timehomogeneity} is nonconstant under the stated assumptions, note first that the right hand side of \eqref{eq:timehomogeneity} is zero at $\beta = 0$. If the left hand side is nonzero, the polynomial is nonconstant. If the left hand side is zero, then the rank condition \eqref{eq:FHregularitycondition} ensures that the derivative of the right hand side is nonzero at $\beta=0$, so that the right hand side, and thus the polynomial, is nonconstant. 
\end{proof}
Rank condition \eqref{eq:FHregularitycondition} adapts \eqref{eq:rank} to the nonstationary case. Unlike the stationary dynamic choice problem, the nonstationary problem does not require that the discount factor lies in $[0,1)$. We leave the definition of the domain of the discount factor to the reader. 

In a study of identification in nonstationary models, \cite{arcidiaconomillerlongshort} distinguished between identification in long panels, which include the terminal period, and short panels, which do not. In general, Theorem \ref{th:finitehorizon} requires long panels. However, for models with $\rho$-period dependence, it also applies to short panels that extend to at least period $t+\rho$. For instance, in Zurcher's renewal problem with a finite horizon, mileage is still single action ($K$) one-period dependent,  so that the discount factor can be point identified in short panels until period $t+1$.

\section{Empirical content}\label{s:empiricalcontent}

The previous section focused on identification and gave conditions under which the primitives can be recovered from the data. In applications, we need to entertain the possibility that the model is misspecified and did not generate the data to begin with. It is well known that the unrestricted model has no empirical content: It can rationalize any choice data $\{\mathbf{p}_k,\mathbf{Q}_k;k\in{\cal D}\}$. This section shows that the model under exclusion restrictions can be rejected by data. 

The standard result for the unrestricted stationary model follows from a version of \citeauthor{ecta02:magnacthesmar}'s Proposition 2: For any given data $\{\mathbf{p}_k,\mathbf{Q}_k;k\in{\cal D}\}$, $\mathbf{u}_K=\mathbf{0}$, and $\beta\in[0,1)$, there exists a unique set of primitive utilities $\{\mathbf{u}_k,k\in{\cal D}/\{K\}\}$ that rationalizes the data. Specifically, $\mathbf{m}=-\ln\mathbf{p}_K$. Then, $\mathbf{v}_K$ follows from $\mathbf{u}_K=\mathbf{0}$ and \eqref{eq:vK}. Next, by \eqref{eq:choiceprobinversion},  $\mathbf{v}_k=\mathbf{v}_K+\ln\mathbf{p}_k-\ln\mathbf{p}_K$ for $k\in{\cal D}/\{K\}$ ensures that the value functions are compatible with the choice probability data. In turn, by \eqref{eq:bellmanprobs}, these value functions are uniquely generated by the primitive utilities $\mathbf{u}_k=\mathbf{v}_k-\beta \mathbf{Q}_k\left[\mathbf{m}+\mathbf{v}_K\right]$ for $k\in{\cal D}/\{K\}$ (note that $\mathbf{v}_K$ was already set to be consistent with $\mathbf{u}_K=\mathbf{0}$). 

This result justifies our focus on the identification of the discount factor $\beta$ in the previous section: 
Once the discount factor is identified, we can find unique primitive utilities that rationalize the data. 
The empirical consequences of a violation of the assumed exclusion restriction can manifest themselves in two distinct ways. 

First, in some cases, it may be possible to find primitives that satisfy the false exclusion restriction. If so, these primitives will in general not equal the true primitives. Because we can find primitive utilities that rationalize the data for any discount factor, the data can be of no help in determining the right restriction in this case. Instead, we need to argue for the identifying assumption on other grounds. 

Second, there may not exist discount factors in their domain that are compatible with the data under the assumed exclusion restriction. The subset of the possible data that can be rationalized under an exclusion restriction can be very small. For instance, in a binary choice model with $J = 2$ and  $u^*_1(x_1) = u^*_1(x_2)$, the model cannot generate any state-dependent value contrasts. It follows that this model cannot rationalize any state-dependent choice data. In empirical practice, this may force parameter estimates to lie outside their theoretical domains. In turn, this may lead researchers to statistically reject the model and conclude that at least one of its assumptions is violated. While some solution methods, such as typical nested fixed point algorithms, impose the restriction that $\beta \in [0,1)$, it is easy to use the moment conditions in \eqref{eq:diffAD} for model testing as their computation do not restrict the values $\beta$ can take.

The empirical content of the identified model also gives some scope to test nonnested identifying assumptions against each other. For example, the data in Example \ref{ex:needNoMT} cannot be rationalized under \citeauthor{ecta02:magnacthesmar}'s current value restriction, but are consistent with an exclusion restriction on primitive utility.  Conversely, it is easy to construct data that are inconsistent with the primitive utility restriction, yet can be rationalized by primitives that satisfy the current value restriction.
\begin{example}
\label{ex:moreEmpiricalContent}
Figure \ref{fig:moreEmpiricalContent} displays the left and right hand sides of \citeauthor{ecta02:magnacthesmar}'s moment condition in \eqref{eq:diffMT} and ours in \eqref{eq:diffAD}. There is a $\beta\in[0,1)$ that solves  \eqref{eq:diffMT}, but the moment condition in  \eqref{eq:diffAD} cannot  be met. Intuitively, the increasingly negative contribution of the second (value of choice $K$) term in the right hand side of \eqref{eq:diffAD} limits the possible log choice probability ratio response to the change in states to a level below the observed response. 
\end{example}
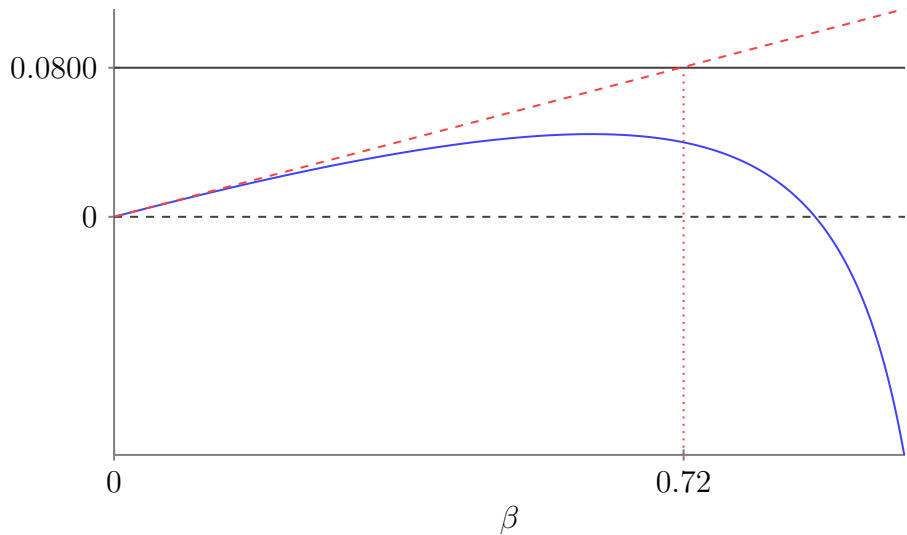
\begin{figure}[t]
\caption{Example of Data that are Consistent with an Exclusion Restriction on Current Values but Not with One on Primitive Utility\label{fig:moreEmpiricalContent}}
\medskip
\centering
\begin{tikzpicture}
   \begin{axis}[width=\figwidth,height=\figheight,
        name=llherr,
        xlabel={$\beta$},
        axis x line=bottom,
        every outer x axis line/.append style={-,color=gray,line width=0.75pt},
        axis y line=left,
        every outer y axis line/.append style={-,color=gray,line width=0.75pt},
        xtick={0,\moreEmpiricalContentPreciseBetaMT},
        xticklabels={$0$,$\moreEmpiricalContentBetaMT$},
        ytick={0,\moreEmpiricalContentPreciseLhs},
        yticklabels={$0$,$\moreEmpiricalContentLhs$},
        every tick/.append style={line width=0.75pt},
	xmin=0, 
	xmax=1,
	ymin=\moreEmpiricalContentMinRhs,
        scaled ticks=false,
        /pgf/number format/precision=2,
        /pgf/number format/set thousands separator={}]
        \draw[color=black!75,line width=0.75pt,solid] (axis cs:0,\moreEmpiricalContentPreciseLhs) -- (axis cs:1,\moreEmpiricalContentPreciseLhs);
        \draw[color=black!75,line width=0.75pt,dashed] (axis cs:0,0) -- (axis cs:1,0);
        \addplot[color=blue!75,smooth,mark=,line width=0.75pt] table[x=beta,y=rhs,col sep=comma]{matlab/data/moreEmpiricalContent.csv} ;
        	\addplot[color=red!75,dashed,mark=,line width=0.75pt] table[x=beta,y=rhsMT,col sep=comma]{matlab/data/moreEmpiricalContent.csv} ;
        	\draw[color=red!75,line width=0.75pt,dotted] (axis cs:\moreEmpiricalContentPreciseBetaMT,\moreEmpiricalContentMinRhs) -- (axis cs:\moreEmpiricalContentPreciseBetaMT,\moreEmpiricalContentPreciseLhs);
\end{axis}
\end{tikzpicture}
\medskip
\begin{minipage}{\figwidth}
{\scriptsize Note: For $J=\moreEmpiricalContentJ$ states, $K=\moreEmpiricalContentK$ choices, $k=l=1$, $\tilde{x}_1=x_1$, and $\tilde{x}_2=x_2$, this graph plots the left hand side of  (\ref{eq:diffMT})  and (\ref{eq:diffAD}) (solid black horizontal line) and the right hand sides of (\ref{eq:diffMT})  (dashed red line) and (\ref{eq:diffAD}) (solid blue curve) as functions of $\beta$. The data are $\mathbf{Q}_1(\tilde{x}_1)=\moreEmpiricalContentQiRowOne$, $\mathbf{Q}_1(\tilde{x}_2)=\moreEmpiricalContentQiRowTwo$,}

{\scriptsize 
\vspace*{-3ex} 
\[\mathbf{Q}_K=\moreEmpiricalContentQK\text{, }\mathbf{p}_1=\moreEmpiricalContentPOne\text{, and }\mathbf{p}_K=\moreEmpiricalContentPK.\] 

Consequently, the left hand side of (\ref{eq:diffMT})  and (\ref{eq:diffAD}) equals $\ln\left(p_1(x_1)/p_K(x_1)\right) - \ln\left(p_1(x_2)/p_K(x_2)\right)= \moreEmpiricalContentLhs$. Moreover,  $\mathbf{m}'=\moreEmpiricalContentMPrime$ and $\mathbf{Q}_1(x_1)-\mathbf{Q}_K(x_1)- \mathbf{Q}_1(x_2)+\mathbf{Q}_K(x_2)=\moreEmpiricalContentDeltaQ$, so that the slope of the dashed red line equals $\left[\mathbf{Q}_1(x_1)-\mathbf{Q}_K(x_1)- \mathbf{Q}_1(x_2)+\mathbf{Q}_K(x_2)\right] \mathbf{m}=\moreEmpiricalContentDeltaQM$. A unique value of $\beta$,  $\moreEmpiricalContentBetaMT$, solves  (\ref{eq:diffMT}), but (\ref{eq:diffAD}) has no solution.}
\end{minipage}
\end{figure}
\noindent In practice, we can easily establish whether given data are consistent with one exclusion restriction or the other by verifying whether the corresponding moment condition, \eqref{eq:diffMT} or \eqref{eq:diffAD}, or its empirical analog has a solution $\beta\in[0,1)$. We can formally test either exclusion restriction with a test of the null hypothesis that $\beta \in[0,1)$. 

Finally, the empirical content of the nonstationary model depends on the chosen domain of the discount factor. Therefore, we limit our discussion of this model's empirical content to noting that Theorem \ref{th:finitehorizon} does not guarantee a real root (and less so one in a specified domain for $\beta$) for general choice and state probabilities.

\section{Multiple exclusion restrictions and inference}
\label{s:multiple}

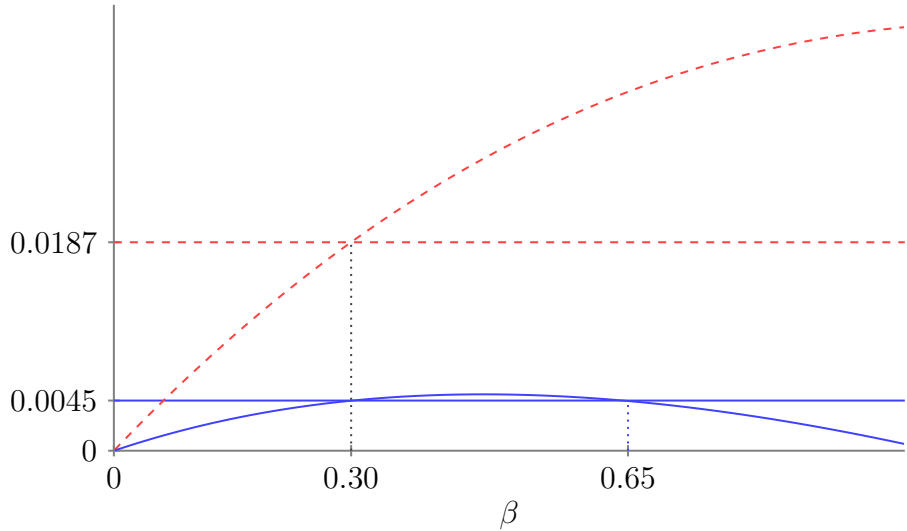
\begin{figure}[t]
\caption{Example with Two Moment Conditions of Which One Identifies the Discount Factor\label{fig:oneRedundant}}
\centering
\medskip
\begin{tikzpicture}
   \begin{axis}[width=\figwidth,height=\figheight,
        name=llherr,
        xlabel={$\beta$},
        axis x line=bottom,
        every outer x axis line/.append style={-,color=gray,line width=0.75pt},
        axis y line=left,
        every outer y axis line/.append style={-,color=gray,line width=0.75pt},
        xtick={0,\oneRedundantOnePreciseBetaOne,\oneRedundantTwoPreciseBetaTwo},
        xticklabels={$0$,$\oneRedundantOneBetaOne$,$\oneRedundantTwoBetaTwo$},
        ytick={0,\oneRedundantOnePreciseLhs,\oneRedundantTwoPreciseLhs},
        yticklabels={$0$,$\oneRedundantOneLhs$,$\oneRedundantTwoLhs$},
        every tick/.append style={line width=0.75pt},
	xmin=0, 
	xmax=1,
	ymin=0,
	ymax=0.04,
        scaled ticks=false,
        /pgf/number format/precision=2,
        /pgf/number format/set thousands separator={}]
        \draw[color=red!75,line width=0.75pt,dashed] (axis cs:0,\oneRedundantOnePreciseLhs) -- (axis cs:1,\oneRedundantOnePreciseLhs);
        \draw[color=blue!75,line width=0.75pt,solid] (axis cs:0,\oneRedundantTwoPreciseLhs) -- (axis cs:1,\oneRedundantTwoPreciseLhs);
        	\draw[color=black!75,line width=0.75pt,dotted] (axis cs:\oneRedundantOnePreciseBetaOne,0) -- (axis cs:\oneRedundantOnePreciseBetaOne,\oneRedundantOnePreciseLhs);
        	\draw[color=blue!75,line width=0.75pt,dotted] (axis cs:\oneRedundantTwoPreciseBetaTwo,0) -- (axis cs:\oneRedundantTwoPreciseBetaTwo,\oneRedundantTwoPreciseLhs);
        \addplot[color=red!75,smooth,mark=,line width=0.75pt,dashed] table[x=beta,y=oneRhs,col sep=comma]{matlab/data/oneRedundant.csv};
        \addplot[color=blue!75,smooth,mark=,line width=0.75pt,solid] table[x=beta,y=twoRhs,col sep=comma]{matlab/data/oneRedundant.csv};
\end{axis}
\end{tikzpicture}
\medskip

\begin{minipage}{\figwidth}
{\scriptsize Note:  For $J=\oneRedundantJ$ states, $K=\oneRedundantK$ choices, and $k=l=1$, this graph plots the left (horizontal lines) and right hand sides (curves) of (\ref{eq:diffAD}) as functions of $\beta$, for $\tilde x_1=x_1$ and $\tilde x_2=x_2$ (corresponding to $u_1(x_1) = u_1(x_2)$; dashed red) and $\tilde x_1=x_3$ and $\tilde x_2=x_4$ (corresponding to $u_1(x_3) = u_1(x_4)$; solid blue). The data are}

{\scriptsize 
\vspace*{-3ex} 
\[\mathbf{Q}_1 = \oneRedundantQi, \mathbf{Q}_K=\oneRedundantQK\text{, }\]
\[\mathbf{p}_1^\prime=\oneRedundantPOne \text{, and }\mathbf{p}_K^\prime=\oneRedundantPK.\] 

Consequently, the left hand sides of (\ref{eq:diffAD}) equal $\ln\left(p_1(x_1)/p_K(x_1)\right) - \ln\left(p_1(x_2)/p_K(x_2)\right)= \oneRedundantOneLhs$ and $\ln\left(p_1(x_3)/p_K(x_3)\right) - \ln\left(p_1(x_4)/p_K(x_4)\right)= \oneRedundantTwoLhs$. A unique value of $\beta$, $\oneRedundantOneBetaOne$, solves (\ref{eq:diffAD}) for $\tilde x_1=x_1$ and $\tilde x_2=x_2$  (dashed red). Two values of $\beta$ solve  (\ref{eq:diffAD}) for $\tilde x_1=x_3$ and $\tilde x_2=x_4$ (solid blue), of which one coincides with the solution to the first moment condition.}
\end{minipage}
\end{figure}

Often, more than one exclusion restriction is available. In particular, economic intuition for an exclusion restriction across states typically suggests the exclusion of a state {\em variable} from the utility function. For example, the state variable $x$ can be partitioned as $(y,z)$, where $z$ does not affect utilities:  $u_k(\tilde y,\tilde z_1)=u_k(\tilde y,\tilde z_2)$ for all $k\in{\cal D}/\{K\}$, $\tilde y$, $\tilde z_1$, and $\tilde z_2>\tilde z_1$.\footnote{We provide a more formal statement of the exclusion of state variables in our discussion of \citeauthor{ier15:fangwang} in \citet{fw19:abbringdaljord}.} This typically gives multiple exclusion restrictions like \eqref{eq:exclusionrestriction}. For example,  if choices, $y$, and $z$ are all binary, we have two exclusion restrictions, one for each possible value of $y$. 

With multiple exclusion restrictions, point identification can be obtained even if each individual moment condition set identifies $\beta$. We give two examples of identification with two exclusion restrictions.
\begin{example}
\label{ex:oneRedundant}
In Figure \ref{fig:oneRedundant}, the moment condition represented by the solid blue line and curve and the one in red dashes have two and one solutions, respectively. Both moment conditions are consistent with a discount factor of $\oneRedundantOneBetaOne$, while the solid moment condition is also consistent with a discount factor of $\oneRedundantTwoBetaTwo$. The dashed moment condition by itself point identifies the discount factor, while the solid moment condition only set identifies it. In this case, the solid moment condition is redundant for point identification. 
\end{example}
\begin{example}
\label{ex:twiceTwo}
In Figure \ref{fig:twiceTwo}, the dashed red moment condition holds for discount factors $\twiceTwoOneBetaOne$ and $\twiceTwoOneBetaTwo$, while the solid blue moment condition is solved by discount factors $\twiceTwoTwoBetaOne$ and $\twiceTwoTwoBetaTwo$. Each individual moment condition is consistent with two discount factors, but only one discount factor solves both moment conditions. 
\end{example}
With choice and transition probabilities generated from a model that satisfies two (or more) exclusion restrictions, the implied two (or more) moment conditions will always share one solution, the discount factor that was used to generate the data. We conjecture that, generically, the moments will not share any further solutions, because different choice and transition probabilities, which vary freely with the primitive utilities, enter the various moment conditions.  

Generic point identification is of limited practical value in our context. First, we are not able to a priori characterize the subset of the model space on which point identification fails in terms of economic concepts. Though this subset is small, it may, for all we know, contain economically important models.\footnote{For example,  \cites{jpe04:ekelandetal} generic identification result for the hedonic model is particularly instructive because it shows that identification fails exactly for the linear-quadratic special case that is at the center of most applied work.}  

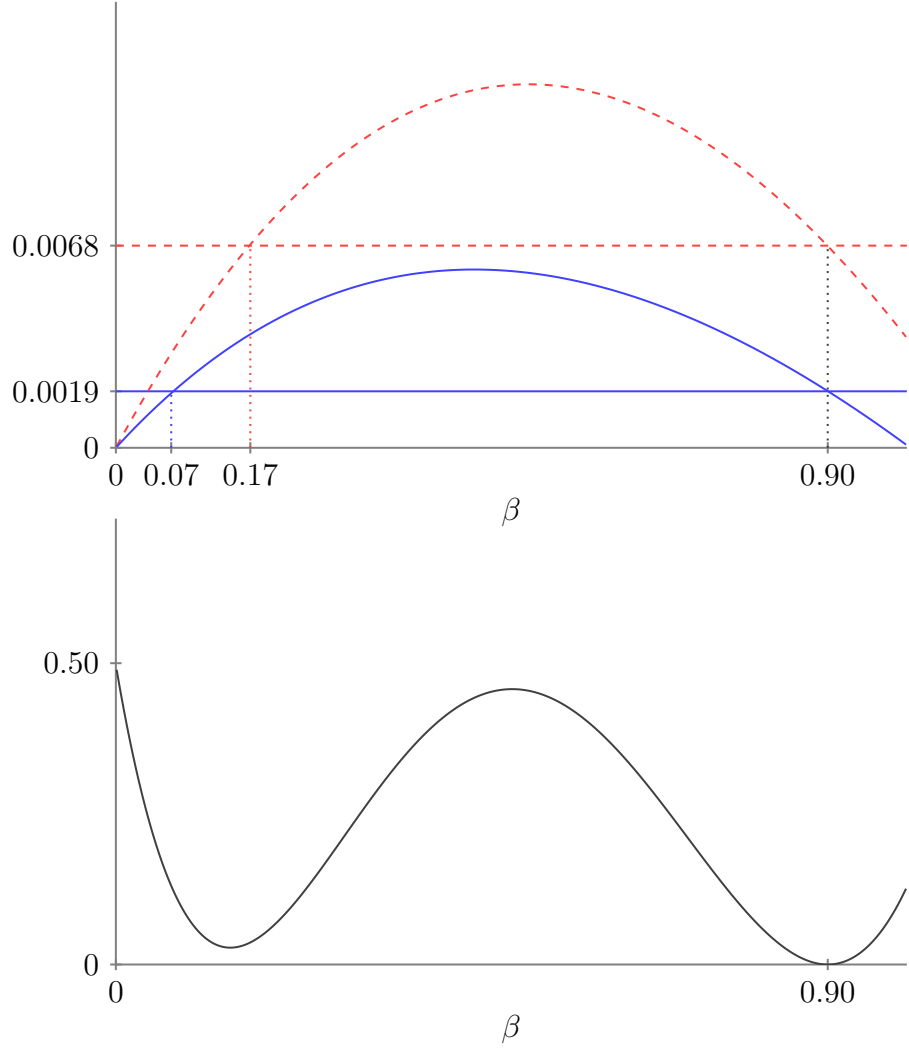
\begin{figure}[t]
\caption{Example with Two Moment Conditions that Jointly Identify the Discount Factor but Individually Do Not\label{fig:twiceTwo}}
\centering
\medskip
\begin{tikzpicture}
   \begin{axis}[width=\figwidth,height=\figheight,
        name=mom,
        xlabel={$\beta$},
        axis x line=bottom,
        every outer x axis line/.append style={-,color=gray,line width=0.75pt},
        axis y line=left,
        every outer y axis line/.append style={-,color=gray,line width=0.75pt},
        xtick={0,\twiceTwoOnePreciseBetaOne,\twiceTwoOnePreciseBetaTwo,\twiceTwoTwoPreciseBetaOne},
        xticklabels={$0$,$\twiceTwoOneBetaOne$,$\twiceTwoOneBetaTwo$,$\twiceTwoTwoBetaOne$},
        ytick={0,\twiceTwoOnePreciseLhs,\twiceTwoTwoPreciseLhs},
        yticklabels={$0$,$\twiceTwoOneLhs$,$\twiceTwoTwoLhs$},
        every tick/.append style={line width=0.75pt},
	xmin=0, 
	xmax=1,
	ymin=0,
	ymax=0.015,
        scaled ticks=false,
        /pgf/number format/precision=2,
        /pgf/number format/set thousands separator={}]
        \draw[color=red!75,line width=0.75pt,dashed] (axis cs:0,\twiceTwoOnePreciseLhs) -- (axis cs:1,\twiceTwoOnePreciseLhs);
        \draw[color=blue!75,line width=0.75pt,solid] (axis cs:0,\twiceTwoTwoPreciseLhs) -- (axis cs:1,\twiceTwoTwoPreciseLhs);
        	\draw[color=black!75,line width=0.75pt,dotted] (axis cs:\twiceTwoOnePreciseBetaTwo,0) -- (axis cs:\twiceTwoOnePreciseBetaTwo,\twiceTwoOnePreciseLhs);
        	\draw[color=red!75,line width=0.75pt,dotted] (axis cs:\twiceTwoOnePreciseBetaOne,0) -- (axis cs:\twiceTwoOnePreciseBetaOne,\twiceTwoOnePreciseLhs);
        	\draw[color=blue!75,line width=0.75pt,dotted] (axis cs:\twiceTwoTwoPreciseBetaOne,0) -- (axis cs:\twiceTwoTwoPreciseBetaOne,\twiceTwoTwoPreciseLhs);
        \addplot[color=red!75,smooth,mark=,line width=0.75pt,dashed] table[x=beta,y=oneRhs,col sep=comma]{matlab/data/twiceTwo.csv};
        \addplot[color=blue!75,smooth,mark=,line width=0.75pt,solid] table[x=beta,y=twoRhs,col sep=comma]{matlab/data/twiceTwo.csv};
\end{axis}

\path(mom.origin) ++(0,-2.7in) coordinate (crit plot position);
\begin{axis}[width=\figwidth,height=\figheight,
        name=crit,
        at = {(crit plot position)},
        xlabel={$\beta$},
        axis x line=bottom,
        every outer x axis line/.append style={-,color=gray,line width=0.75pt},
        axis y line=left,
        every outer y axis line/.append style={-,color=gray,line width=0.75pt},
        xtick={0,\twiceTwoOnePreciseBetaTwo},
        xticklabels={$0$,$\twiceTwoOneBetaTwo$},
        ytick={0,0.5},
        yticklabels={$0$,$0.50$},
        every tick/.append style={line width=0.75pt},
	xmin=0, 
	xmax=1,
	ymin=0,
	ymax=0.74,
        scaled ticks=false,
        /pgf/number format/precision=2,
        /pgf/number format/set thousands separator={}]
        \addplot[color=black!75,smooth,mark=,line width=0.75pt,solid] table[x=beta,y=quadDist,col sep=comma]{matlab/data/twiceTwo.csv};             
\end{axis}
\end{tikzpicture}
\medskip

\begin{minipage}{\figwidth}
{\scriptsize 
Note: For $J=\twiceTwoJ$ states, $K=\twiceTwoK$ choices, and $k=l=1$, the graph in the top panel plots the left (horizontal lines) and right hand sides (curves) of (\ref{eq:diffAD}) as functions of $\beta$, for $\tilde x_1=x_1$ and $\tilde x_2=x_2$ (corresponding to $u_1(x_1) = u_1(x_2)$; dashed red) and $\tilde x_1=x_3$ and $\tilde x_2=x_4$ (corresponding to $u_1(x_3) = u_1(x_4)$; solid blue). The graph in the bottom panel plots the corresponding squared Euclidian distance between the left and right hand sides of (\ref{eq:diffAD}) as a function of $\beta$ (in multiples of $10^{-\twiceTwoNoiseDistScaler}$). The data are}

{\scriptsize 
\vspace*{-3ex}
\[\mathbf{Q}_1 = \twiceTwoQi,~ \mathbf{Q}_K=\twiceTwoQK,\]
\[\mathbf{p}_1^\prime=\twiceTwoPOne\text{, and }\mathbf{p}_K^\prime=\twiceTwoPK.\]

Consequently, the left hand sides of (\ref{eq:diffAD}) equal $\ln\left(p_1(x_1)/p_K(x_1)\right) - \ln\left(p_1(x_2)/p_K(x_2)\right)= \twiceTwoOneLhs$ and $\ln\left(p_1(x_3)/p_K(x_3)\right) - \ln\left(p_1(x_4)/p_K(x_4)\right)= \twiceTwoTwoLhs$. A unique value of $\beta$, $\twiceTwoOneBetaTwo$, solves  (\ref{eq:diffAD}) for both $\tilde x_1=x_1$ and $\tilde x_2=x_2$ (dashed red) and $\tilde x_1=x_3$ and $\tilde x_2=x_4$ (solid blue). In addition, each of these two moment conditions has one other solution.}
\end{minipage}
\end{figure}

Second, we may not learn whether the discount factor is point or set identified in finite samples. While finding the shared solutions to multiple moment conditions is easy if we know the population choice and transition probabilities, locating the shared solutions in finite samples can be difficult due to sampling variation. This suggests that we do not insist on point identification, but accept set identification and use a consistent estimator of the identified set, which may contain one or more points. Set estimators are easy to implement for single parameter problems. We give one example.
\begin{example}
Suppose the population moment conditions are as given in Figure \ref{fig:twiceTwo}. Though each individual moment condition is equally consistent with one small discount factor, at $\twiceTwoTwoBetaOne$ and $\twiceTwoOneBetaOne$, respectively,  and one large discount factor at the true value of $\twiceTwoOnePreciseBetaTwo$, only the latter is a common solution to both moment conditions. The discount factor is therefore point identified in this population.  

\begin{figure}[t]
\caption{Example with Two Moment Conditions that Jointly Identify the Discount Factor but Individually Do Not, Using Noisy Estimates of the Choice Probabilities\label{fig:twiceTwoNoise}}
\centering
\medskip
\begin{tikzpicture}
   \begin{axis}[width=\figwidth,height=\figheight,
        name=noisyMom,
        xlabel={$\beta$},
        axis x line=bottom,
        every outer x axis line/.append style={-,color=gray,line width=0.75pt},
        axis y line=left,
        every outer y axis line/.append style={-,color=gray,line width=0.75pt},
        xtick={0,\twiceTwoNoiseOnePreciseBetaOne,\twiceTwoNoiseOnePreciseBetaTwo,\twiceTwoNoiseTwoPreciseBetaOne,\twiceTwoNoiseTwoPreciseBetaTwo},
        xticklabels={$0$,$\twiceTwoNoiseOneBetaOne$,$\twiceTwoNoiseOneBetaTwo$,$\twiceTwoNoiseTwoBetaOne$,$\twiceTwoNoiseTwoBetaTwo$},
        ytick={0,\twiceTwoNoiseOnePreciseLhs,\twiceTwoNoiseTwoPreciseLhs},
        yticklabels={$0$,$\twiceTwoNoiseOneLhs$,$\twiceTwoNoiseTwoLhs$},
        every tick/.append style={line width=0.75pt},
	xmin=0, 
	xmax=1,
	ymin=0,
	ymax=0.015,
        scaled ticks=false,
        /pgf/number format/precision=2,
        /pgf/number format/set thousands separator={}]
        \draw[color=red!75,line width=0.75pt,dashed] (axis cs:0,\twiceTwoNoiseOnePreciseLhs) -- (axis cs:1,\twiceTwoNoiseOnePreciseLhs);
        \draw[color=blue!75,line width=0.75pt,solid] (axis cs:0,\twiceTwoNoiseTwoPreciseLhs) -- (axis cs:1,\twiceTwoNoiseTwoPreciseLhs);
        	\draw[color=red!75,line width=0.75pt,dotted] (axis cs:\twiceTwoNoiseOnePreciseBetaOne,0) -- (axis cs:\twiceTwoNoiseOnePreciseBetaOne,\twiceTwoNoiseOnePreciseLhs);	
	\draw[color=red!75,line width=0.75pt,dotted] (axis cs:\twiceTwoNoiseOnePreciseBetaTwo,0) -- (axis cs:\twiceTwoNoiseOnePreciseBetaTwo,\twiceTwoNoiseOnePreciseLhs);
        	\draw[color=blue!75,line width=0.75pt,dotted] (axis cs:\twiceTwoNoiseTwoPreciseBetaOne,0) -- (axis cs:\twiceTwoNoiseTwoPreciseBetaOne,\twiceTwoNoiseTwoPreciseLhs);
        	\draw[color=blue!75,line width=0.75pt,dotted] (axis cs:\twiceTwoNoiseTwoPreciseBetaTwo,0) -- (axis cs:\twiceTwoNoiseTwoPreciseBetaTwo,\twiceTwoNoiseTwoPreciseLhs);
        \addplot[color=red!75,smooth,mark=,line width=0.75pt,dashed] table[x=beta,y=oneRhs,col sep=comma]{matlab/data/twiceTwoNoise.csv};
        \addplot[color=blue!75,smooth,mark=,line width=0.75pt,solid] table[x=beta,y=twoRhs,col sep=comma]{matlab/data/twiceTwoNoise.csv};
\end{axis}

\path(noisyMom.origin) ++(0,-2.7in) coordinate (noisyCrit plot position);
\begin{axis}[width=\figwidth,height=\figheight,
        name=noisyCrit,
        at = {(noisyCrit plot position)},
        xlabel={$\beta$},
        axis x line=bottom,
        every outer x axis line/.append style={-,color=gray,line width=0.75pt},
        axis y line=left,
        every outer y axis line/.append style={-,color=gray,line width=0.75pt},
        xtick={0,\twiceTwoNoisePreciseBetaSetOne,\twiceTwoNoisePreciseBetaSetTwo,\twiceTwoNoisePreciseBetaSetThree,\twiceTwoNoisePreciseBetaSetFour},
        xticklabels={$0$,$\twiceTwoNoiseBetaSetOne$,$\twiceTwoNoiseBetaSetTwo$,$\twiceTwoNoiseBetaSetThree$,$\twiceTwoNoiseBetaSetFour$},
        ytick={0,\twiceTwoNoisePreciseCriticalValue,0.5},
        yticklabels={$0$,$s_n$,$0.50$},
        every tick/.append style={line width=0.75pt},
	xmin=0, 
	xmax=1,
	ymin=0,
	ymax=0.74,
        scaled ticks=false,
        /pgf/number format/precision=2,
        /pgf/number format/set thousands separator={}]
        \draw[color=black!75,line width=0.75pt,solid] (axis cs:0,\twiceTwoNoisePreciseCriticalValue) -- (axis cs:1,\twiceTwoNoisePreciseCriticalValue);
        \addplot[color=black!75,smooth,mark=,line width=0.75pt,solid] table[x=beta,y=quadDist,col sep=comma]{matlab/data/twiceTwoNoise.csv};             
	\addplot [line width=0.75pt,color=green!75,mark=,fill=green, fill opacity=0.05]coordinates {
		(\twiceTwoNoisePreciseBetaSetOne, 0) 
		(\twiceTwoNoisePreciseBetaSetTwo, 0) 
            	(\twiceTwoNoisePreciseBetaSetTwo,\twiceTwoNoisePreciseCriticalValue)  
           	(\twiceTwoNoisePreciseBetaSetOne,\twiceTwoNoisePreciseCriticalValue)
                 (\twiceTwoNoisePreciseBetaSetOne, 0) };
	\addplot [line width=0.75pt,color=green!75,mark=,fill=green, fill opacity=0.05]coordinates {
		(\twiceTwoNoisePreciseBetaSetThree, 0) 
		(\twiceTwoNoisePreciseBetaSetFour, 0) 
            	(\twiceTwoNoisePreciseBetaSetFour,\twiceTwoNoisePreciseCriticalValue)  
           	(\twiceTwoNoisePreciseBetaSetThree,\twiceTwoNoisePreciseCriticalValue)
                 (\twiceTwoNoisePreciseBetaSetThree, 0) };
\end{axis}
\end{tikzpicture}
\medskip

\begin{minipage}{\figwidth}
{\scriptsize Note: This figure redraws Figure \ref{fig:twiceTwo} for the same values of $\mathbf{Q}_1$ and $\mathbf{Q}_K$, but randomly perturbed values of its choice probabilities $\mathbf{p}_1$ and $\mathbf{p}_K$. Rounded to two digits, the perturbed choice probabilities equal those reported below Figure \ref{fig:twiceTwo}. Consequently, the perturbation to $\mathbf{m}=-\ln\mathbf{p}_K$ is very small too, so that  the right hand sides of  (\ref{eq:diffAD}) are very close to those plotted in Figure \ref{fig:twiceTwo}. The left hand sides of (\ref{eq:diffAD}), however, now equal $\ln\left(p_1(x_1)/p_K(x_1)\right) - \ln\left(p_1(x_2)/p_K(x_2)\right)=\twiceTwoNoiseOneLhs$ (instead of $\twiceTwoOneLhs$) and $\ln\left(p_1(x_3)/p_K(x_3)\right) - \ln\left(p_1(x_4)/p_K(x_4)\right)= \twiceTwoNoiseTwoLhs$ (instead of $\twiceTwoTwoLhs$). The resulting moment conditions again have two solutions. However, they no longer share a common solution and the squared Euclidian distance in the bottom panel never attains zero. The green shaded areas highlight the intervals $[\twiceTwoNoiseBetaSetOne,\twiceTwoNoiseBetaSetTwo]$ and $[\twiceTwoNoiseBetaSetThree,\twiceTwoNoiseBetaSetFour]$ of values of $\beta$ at which the distance is below some critical level $s_n$ (which is taken to be $\twiceTwoNoiseCriticalValue\times 10^{-\twiceTwoNoiseDistScaler}$ in this example).}
\end{minipage}
\end{figure}
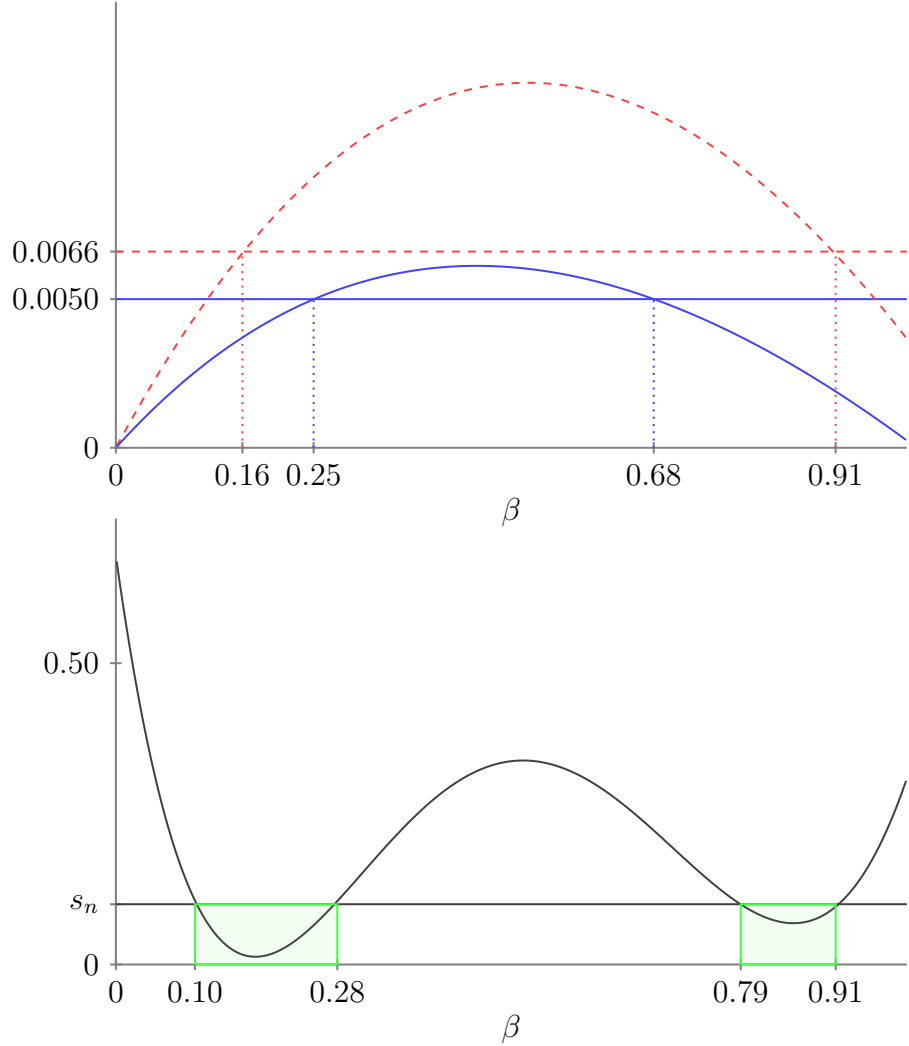

In the top panel of Figure \ref{fig:twiceTwoNoise}, the same two moment conditions are plotted with sampling variation in the choice data. One sample moment condition is solved by discount factors $\twiceTwoNoiseOnePreciseBetaOne$ and $\twiceTwoNoiseOnePreciseBetaTwo$ and the other by discount factors $\twiceTwoNoiseTwoPreciseBetaOne$ and $\twiceTwoNoiseTwoPreciseBetaTwo$. The data do not clearly reveal that the point-identified true discount factor is $\twiceTwoOnePreciseBetaTwo$. If anything, the data suggest point identification in the lower region. Even if point identification cannot be determined a priori without further assumptions, the discount factor is still set identified and we can use consistent set estimators. 

Following \cite{chernozhukov2007estimation} and \cite{romanoshaikh10}, suppose that the identified set ${\cal B}=\{\beta\in[0,1):S(\beta)=0\}$ for some population criterion function $S:[0,1)\rightarrow[0,\infty)$. Note that we can alternatively write ${\cal B} = \arg \min_{\beta \in [0,1)} S(\beta)$. This suggests that we estimate ${\cal B}$ by a random contour set ${\cal C}_n(s)= \{\beta \in [0,1): a_n S_n(\beta) \leq s\}$ for some level $s>0$ and normalizing sequence $\{a_n\}$, where $S_n(\beta)$ is the sample equivalent of $S(\beta)$ and $n$ is the sample size. For a given confidence level $\alpha\in(0,1)$, $s$ is set to equal a consistent estimator $s_n$ of the $\alpha$-quantile of $\sup_{\beta \in {\cal B}}  a_n S_n(\beta)$, so that the estimator ${\cal C}_n(s_n)$ asymptotically contains the identified set with probability $\alpha$:
\[
\lim_{n \rightarrow \infty} \Pr\{{\cal B} \subseteq {\cal C}_n(s_n)\} = \alpha. 
\]

The bottom panel of Figure \ref{fig:twiceTwoNoise} illustrates one such estimator. The criterion $S_n(\beta)$ is here a quadratic form in the difference between the left and right hand sides of \eqref{eq:diffAD} evaluated at consistent estimators of the choice and transition probabilities using equal weights. The critical value $s_n$ is given as the horizontal line. The estimated set is ${\cal C}_n(s_n) = [\twiceTwoNoisePreciseBetaSetOne, \twiceTwoNoisePreciseBetaSetTwo] \cup[\twiceTwoNoisePreciseBetaSetThree, \twiceTwoNoisePreciseBetaSetFour]$. The data are equally consistent with a range of small discount factors and a range of large discount factors, but an intermediate range $(\twiceTwoNoisePreciseBetaSetTwo, \twiceTwoNoisePreciseBetaSetThree)$ is rejected at the $\alpha$-level, along with discount factors smaller than $\twiceTwoNoisePreciseBetaSetOne$ and larger than $\twiceTwoNoisePreciseBetaSetFour$.

Under some regularity conditions, the set estimator converges to the identified set as the sample size grows. Since the identified set is a point in this example, in the limit, the subset of ${\cal C}_n(s_n)$ with small discount factors vanishes and its subset with large discount factors degenerates to the population discount factor $\twiceTwoOnePreciseBetaTwo$. While these set estimators are computationally demanding for parameter spaces with even just a handful of dimensions, they are easy to implement in a one-dimensional case such as ours. 
\end{example}

\section{Practical considerations}
We conclude with some considerations relevant to applications. If the discount factor is point identified, utilities are as well, and $\beta$ and $\textbf{u}$ can be estimated jointly by standard methods, e.g. maximum likelihood, with the exclusion restriction on $\textbf{u}$ imposed. Standard inference for extremum estimators applies \citep[e.g.][]{nh94:neweymcfadden}. 

Typical implementations of such joint estimators will impose functional form assumptions on the utility function that have identifying power on their own (\citeauthor{qe18:komarovaetal}). Then, it is unclear how much information about the discount factor is carried by the exclusion restrictions, which are economically motivated, and how much is carried by the functional forms, which are typically more arbitrary.  An alternative approach is to use that,  by a version of \citeauthor{ecta02:magnacthesmar}'s Proposition 2 (see Section \ref{s:empiricalcontent}), there exist unique utilities that rationalize the data for any given discount factor. This suggests a two-step estimation procedure. In the first step,  $\beta$ can be recovered from the moment condition in \eqref{eq:diffAD}. In the second step, the utilities are estimated using the moment conditions in \eqref{eq:u}, taking the discount factor recovered in the first step as given. This way, estimation of the discount factor is robust to misspecification of the utility function. See \citet{daljordetal19} for an application of this approach.   

If the discount factor is not known to be point identified, one may construct the sample analogues to \eqref{eq:diffAD} and plot the criterion function, as in the bottom panel of Figure \ref{fig:twiceTwoNoise}. If the criterion function is close to quadratic around a unique minimum on the domain of $\beta$, one may proceed as if the model is point identified. If the criterion function is decidedly nonquadratic, as in Figure \ref{fig:twiceTwoNoise}, then the discount factor can be estimated in a first step using a set estimator of the kind described in Section \ref{s:multiple}. These estimates are a set of possibly intersecting subintervals of $[0,1)$. In a second step, utilities and counterfactual choice probabilities can be computed for each $\beta$ in the identified set.

\clearpage
\appendix
\section*{Appendix}

\section*{Identification with general reference utility}
\label{app:uK}

Consider the stationary model of Section \ref{s:model}. Suppose that we know $\mathbf{u}_K$ up to a constant additive shift; that is, $\mathbf{u}_K=\gamma\mathbf{1}+\bar{\mathbf{u}}_K$,  with $\gamma\in\mathbb{R}$  unknown, $\mathbf{1}$ the $J\times 1$ vector of ones, and $\bar{\mathbf{u}}_K$ a known $J\times 1$ vector with $j$-th element $\bar{u}_K(x_j)$. Then, we can rewrite \eqref{eq:u} as
\begin{equation}
\label{eq:uuK}
\ln \mathbf{p}_k - \ln \mathbf{p}_K
=  \beta\left[\mathbf{Q}_k-\mathbf{Q}_K\right]\left[\mathbf{I}-\beta \mathbf{Q}_K\right]^{-1}\left(\mathbf{m}+\bar{\mathbf{u}}_K\right) + \mathbf{u}_k-\gamma\mathbf{1}-\bar{\mathbf{u}}_K.
\end{equation}
Note that the constant additive shift $\gamma\mathbf{1}$ drops from the first term, which is a difference in expectations under choices $k$ and $K$. 

Now suppose that $u_k^*(\tilde x_1)-u_l^*(\tilde x_2)$ is known, but not necessarily zero,  for some known choices $k\in{\cal D}/\{K\}$ and $l \in {\cal D}$, and known states $\tilde x_1 \in {\cal X}$ and $\tilde x_2\in{\cal X}$; with either $k\neq l$,  $\tilde{x}_1\not =\tilde{x}_2$, or both. This is an exclusion restriction that encompasses \eqref{eq:exclusionrestriction} in the main text as a special case. Under this generalized exclusion restriction, (\ref{eq:uuK}) implies
\begin{equation}
\label{eq:diffADGen}
\begin{split}
&\ln\left(p_k(\tilde x_1)/p_K(\tilde x_1)\right) - \ln\left(p_l(\tilde x_2)/p_K(\tilde x_2)\right)- \Delta^2 u\\
&~~~~~~~~~~~~~~~~=\beta\left[\mathbf{Q}_k(\tilde x_1)-\mathbf{Q}_K(\tilde x_1)- \mathbf{Q}_l(\tilde x_2)+\mathbf{Q}_K(\tilde x_2)\right]  \left[\mathbf{I}-\beta \mathbf{Q}_K\right]^{-1}\bar{\mathbf{m}},
\end{split}
\end{equation}
with $\Delta^2 u\equiv u_k^*(\tilde x_1)-u_l^*(\tilde x_2)-\bar{u}_K(\tilde x_1)+\bar{u}_K(\tilde x_2)$ and $\bar{\mathbf{m}}\equiv\mathbf{m}+\bar{\mathbf{u}}_K$ known. The factor multiplying $\beta$ in the right hand side of \eqref{eq:diffAD} can again be interpreted in terms of incentives related to differences in expected future utilities, which now include the known utilities derived from the reference choice $K$. Multiplying these ``incentives'' by the discount factor $\beta$ gives the log choice probability response, corrected for the known effects of the current utility contrast $\Delta^2 u$, in the left hand side of  \eqref{eq:diffAD}.     

The analysis of the main text applies to this generalization with straightforward adaptations. In particular, \eqref{eq:diffADGen} is a moment condition in only one unknown, the discount factor $\beta$, and can be taken directly to data. The following generalization of Theorem \ref{th:ident} can be proved like that theorem.
\begin{theorem}
\label{th:identGen}
Suppose that $u_k^*(\tilde x_1)-u_l^*(\tilde x_2)$ is known for some $k\in{\cal D}/\{K\}$, $l \in {\cal D}$, $\tilde x_1 \in {\cal X}$, and $\tilde x_2\in{\cal X}$; with either $k \neq l$,  $\tilde{x}_1\not =\tilde{x}_2$, or both. Moreover, suppose that either the left hand side of (\ref{eq:diffADGen}) is nonzero (that is, $p_k(\tilde x_1)/p_K(\tilde x_1)-p_l(\tilde x_2)/p_K(\tilde x_2)\neq\Delta^2 u$)
or a generalization of \citeauthor{ecta02:magnacthesmar}'s rank condition \eqref{eq:rankMT} holds:
\begin{equation*}
\left[\mathbf{Q}_k(\tilde x_1)-\mathbf{Q}_K(\tilde x_1)- \mathbf{Q}_l(\tilde x_2)+\mathbf{Q}_K(\tilde x_2)\right]\bar{\mathbf{m}}\not = 0.
\end{equation*} 
Then, the identified set ${\cal B}$ is a closed discrete subset of $[0,1)$.
\end{theorem}
A version of Corollary \ref{cor:finite} follows directly and so do the simplifications that arise from finite dependence, in particular those that arise in renewal and optimal stopping problems.  Finally, it is easy to adapt the analysis in this appendix to the nonstationary case. We will not pursue that here. 

This appendix (in particular, a comparison of moment conditions \eqref{eq:diffAD} and \eqref{eq:diffADGen}) demonstrates that the analysis in the main text extends
\begin{itemize} 
\item without change to the case in which $u^*_K(x)$ equals a (not necessarily zero or even known) constant;
\item with a simple, known adjustment to the choice probability response in the left hand side of \eqref{eq:diffAD} to the case that $u_k^*(\tilde x_1)-u_l^*(\tilde x_2)$ is known, but not necessarily zero; and
\item with another such adjustment to the  left hand side of \eqref{eq:diffAD} {\em and} a known adjustment to the polynomial in the right hand side of \eqref{eq:diffAD} if $u^*_K$ is only known up to a constant additive shift, but not necessarily constant. 
\end{itemize}
This shows that our analysis can directly be applied to problems in which a state independent reference utility exists (as is typically assumed in applied work) and directly complements results on the identification of more general reference utility specifications.

\clearpage

\pdfbookmark[0]{References}{pdfbm:refs}
\bibliographystyle{chicago}
\bibliography{alljaap2}



\end{document}